\documentclass[11pt]{article}
\usepackage{amsmath,amssymb,amsfonts,amsthm,epsfig,verbatim}
\usepackage{times}
\usepackage{color}
\newif\ifhyper\IfFileExists{hyperref.sty}{\hypertrue}{\hyperfalse}
\hypertrue
\ifhyper\usepackage{hyperref}\fi

\def\confversion{0}
\def\withcolors{1}

\newtheorem{theorem}{Theorem}

\newtheorem{lemma}[theorem]{Lemma}
\newtheorem{proposition}[theorem]{Proposition}
\newtheorem{corollary}[theorem]{Corollary}
\newtheorem{claim}[theorem]{Claim}
\newtheorem{fact}[theorem]{Fact}

\newtheorem{definition}[theorem]{Definition}
\newtheorem{remark}[theorem]{Remark}
\newtheorem{observation}[theorem]{Observation}
\newcommand{\eps}{\epsilon}

\newcommand{\etal}{{\em et al.\ }}

\newcommand{\braket}[1]{\langle #1 \rangle}

\newcommand{\Var}{\operatorname{Var}}
\newcommand{\Cov}{\operatorname{Cov}}

\newcommand{\sign}{\mathrm{sign}}

\newcommand{\ignore}[1]{}

\newcommand{\cref}[1]{Corollary~\ref{cor:#1}}

\newcommand{\bits}{\{-1,1\}}
\newcommand{\bn}{\bits^n}

\newcommand{\R}{{\mathbb{R}}}

\newcommand{\Z}{{\mathbb Z}}
\newcommand{\E}{\operatorname{{\bf E}}}

\newcommand{\littlesum}{\mathop{\textstyle \sum}}

\newcommand{\poly}{\mathrm{poly}}
\newcommand{\polylog}{\mathrm{polylog}}

\newcommand{\Inf}{\mathrm{Inf}}

\newcommand{\eqdef}{\stackrel{\textrm{def}}{=}}

\renewcommand{\Pr}{\operatorname{{\bf Pr}}}

\newcommand{\tr}{\mathrm{tr}}
\newcommand{\dtv}{d_{\mathrm{TV}}}
\newcommand{\dk}{d_{\mathrm K}}
\newcommand{\ci}{{\mathrm{ci}}}

\ifnum\withcolors=1

\newcommand{\new}[1]{{\color{red} {#1}}}

\else

\newcommand{\new}[1]{{ {#1}}}

\fi

\makeatletter
\renewcommand{\subsection}{\@startsection{subsection}{2}{0pt}{-6pt}{-5pt}{\normalsize\bf}}
\renewcommand{\subsubsection}{\@startsection{subsubsection}{3}{0pt}{-12pt}{-5pt}{\normalsize\bf}}
\makeatother

\newcommand{\rnote}[1]{\footnote{{\bf [[Rocco: {#1}\bf ]] }}}
\newcommand{\inote}[1]{\footnote{{\bf [[Ilias: {#1}\bf ]] }}}
\newcommand{\anote}[1]{\footnote{{\bf [[Anindya: {#1}\bf ]] }}}

\newcommand{\depth}{{\mathrm{depth}}}

\newcommand{\T}{{{\cal T}}}

\date{}

\title{
Deterministic Approximate Counting for \\
Degree-$2$ Polynomial
Threshold Functions\\
}

\author{Anindya De\thanks{{\tt anindya@math.ias.edu}.  Research supported by Templeton Foundation Grant 21674 and  NSF CCF-1149843. }\\
Institute for Advanced Study\\
\and
Ilias Diakonikolas\thanks{{\tt ilias.d@ed.ac.uk}. Some of this work was done while the author was at UC Berkeley supported by a Simons Fellowship.} \\
University of Edinburgh\\
\and Rocco A.\ Servedio\thanks{{\tt rocco@cs.columbia.edu}. Supported by NSF grants CNS-0716245, CCF-0915929, and CCF-1115703.}\\
Columbia University\\
}

\begin{document}

\setcounter{page}{0}

\maketitle

\thispagestyle{empty}

\begin{abstract}
We give a {\em deterministic} algorithm for
approximately computing the fraction of Boolean assignments
that satisfy a degree-$2$ polynomial threshold function.
Given a degree-2 input polynomial $p(x_1,\dots,x_n)$
and a parameter $\eps > 0$, the algorithm approximates
\[
\Pr_{x \sim \{-1,1\}^n}[p(x) \geq 0]
\]
to within an additive $\pm \eps$ in time $\poly(n,2^{\poly(1/\eps)})$.
Note that it is NP-hard to determine whether the above probability
is nonzero, so any sort of multiplicative approximation is almost certainly
impossible even for efficient randomized algorithms.
This is the first deterministic algorithm for this counting problem
in which the running time is polynomial in $n$ for $\eps= o(1)$.
For ``regular'' polynomials $p$ (those in which no individual variable's
influence is large compared to the sum of all $n$
variable influences)
our algorithm runs in $\poly(n,1/\eps)$ time.
The algorithm also runs in $\poly(n,1/\eps)$ time to approximate
$\Pr_{x \sim N(0,1)^n}[p(x) \geq 0]$ to within an additive $\pm \eps$,
for any degree-2 polynomial $p$.
\ignore{
 and hence using invariance principles we are also able to achieve $\poly(n,1/\eps)$ time for the Boolean case when the
polynomial $p$ is sufficiently ``regular.''
}

As an application of our counting result, we give a deterministic FPT
multiplicative $(1 \pm \eps)$-approximation algorithm
to approximate the $k$-th absolute moment $\E_{x \sim \{-1,1\}^n}[|p(x)^k|]$
of a degree-2 polynomial.  The algorithm's running time is of
the form
$\poly(n) \cdot f(k,1/\eps)$.

\end{abstract}

\newpage

\newpage

\section{Introduction}

A \emph{degree-$d$ polynomial threshold function} (PTF) is
a Boolean function $f: \{-1,1\}^n \to \{-1,1\}$ defined by
$f(x) = \sign(p(x))$ where $p: \{-1,1\}^n \to \R$ is a degree-$d$
polynomial.  In the special case where $d=1$, degree-$d$ PTFs are often
referred to as \emph{linear threshold functions} (LTFs)
or \emph{halfspaces}.  While LTFs and low-degree
PTFs have been researched for decades (see e.g., \cite{
MyhillKautz:61, MTT:61, MinskyPapert:68, Muroga:71, GHR:92, Orponen:92,
Hastad:94,Podolskii:09} and many other works)
their study has recently received new impetus as they have
played important roles in complexity theory
\cite{Sherstov:08,Sherstov:09,DHK+:10,Kane:10,Kane12,Kane12GL,KRS12},
learning theory \cite{KKMS:08,SSS11,DOSW:11,DDFS:12stoc},
voting theory \cite{APL:07, DDS12icalp} and other areas.

An  important problem associated with LTFs and PTFs is that of deterministically estimating the fraction
of assignments that satisfy a given LTF or PTF over $\{-1,1\}^n$.
In particular, in this paper we are interested in deterministically estimating the fraction of satisfying assignments for PTFs of degree $d=2$.
This problem is motivated by the long line of work on \emph{deterministic approximate counting algorithms},
starting with the seminal work of Ajtai and Wigderson~\cite{AjtaiWigderson:85} who gave non-trivial  deterministic counting algorithms for constant-depth circuits.
 Since then much progress has been made on the design of deterministic counting algorithms  for other classes
 of Boolean functions like DNFs, low-degree $GF[2]$ polynomials and LTFs~\cite{LubyVelickovic:96, GopalanMR:13, Viola09, GKMSVV11}.
 Problems of this sort can be quite challenging;
 after close to three decades of effort, deterministic polynomial time counting algorithms
 are not yet known for a simple class like polynomial-size DNFs.

Looking beyond Boolean functions, there has been significant work on obtaining deterministic
 approximate counting algorithms for combinatorial problems using ideas and techniques from statistical physics. This includes work on counting matchings~\cite{BGKNT07}, independent sets~\cite{Weitz06}, proper colorings~\cite{LLY13} and other problems in statistical physics~\cite{BG08}. We note that there is interest in obtaining such deterministic algorithms despite the fact that in some of these cases an optimal randomized algorithm is already known (e.g., for counting matchings~\cite{JVS:01}) and the performance of the corresponding deterministic algorithm is significantly worse~\cite{BGKNT07}. For this paper, the most relevant prior work are the results of Gopalan \etal and Stefankovic \etal~\cite{GKMSVV11} who independently obtained  deterministic $\poly(n,1/\eps)$ time algorithms for counting the satisfying assignments of an LTF up to $(1\pm \eps)$ multiplicative error.  (As we discuss later, in contrast with LTFs it is NP-hard to count the satisfying assignments of a degree-$d$ PTF for any $d >1$ up to any multiplicative factor.  Thus, the right notion of approximation for degree-$2$ PTFs is additive error.)

There has recently been significant work in the literature on \emph{unconditional derandomization} of LTFs and PTFs. The starting point of these works are the results of Rabani and Shpilka~\cite{RS:08} and Diakonikolas et al \cite{DGJ+09} who gave explicit constructions of polynomial-sized hitting sets and pseudorandom generators for LTFs. Building on these works, Meka and Zuckerman~\cite{MZstoc10} and subsequently Kane~\cite{Kane11ccc, Kane11focs, Kane12subpoly} constructed polynomial-sized PRGs for degree-$d$ PTFs for $d >1$. These PRGs trivially imply deterministic polynomial-time counting algorithms for any fixed $d$ and fixed $\epsilon>0$. While there has been significant research on improving the dependence of the size of these PRGs on $\eps$, the best construction in the current state of the art is  due to Kane~\cite{
Kane12} who gave an explicit PRG for degree-$d$ polynomial
threshold functions over $\{-1,1\}^n$
of size $n^{O_d(1) \cdot \poly(1/\eps)}$.  (In a related but
different line of work \cite{Kane11ccc,Kane11focs,Kane12subpoly}
focusing on PRGs for degree-$d$ PTFs over the Gaussian $\mathcal{N}(0,1)^n$
distribution, the strongest result to date is that of
\cite{Kane12subpoly}
which for any constant degree $d$ gives an explicit PRG of size
$n^{f_d(1/\eps)}$ for degree-$d$ PTFs; here $f_d(1/\eps)$ is
a  slightly sub-polynomial function of $1/\eps$, even for $d=2$). As a consequence, the resulting deterministic counting algorithms have a running time which is at least $n^{O_d(1) \cdot \poly(1/\eps)}$ and thus the running time of these algorithms is not a \emph{fixed} polynomial in $n$. In particular, for any $\eps = o(1)$, the running time of these algorithms is super-polynomial in $n$.

\subsection{ Our contributions.}

In this work we give the first \emph{fixed}
polynomial time deterministic algorithm for a PTF problem of this sort.
As our main result, for all $\eps > 0$
we give a fixed poly$(n)$-time algorithm
to deterministically $\pm \eps$-approximately count the number of
satisfying assignments to a degree-2 PTF:

\begin{theorem} \label{thm:count-boolean-ptf-informal}
[Deterministic approximate counting of degree-2 PTFs over
$\{-1,1\}^n$, informal statement]
There is a deterministic algorithm which, given
a degree-$2$ polynomial $q(x_1,\dots,x_n)$ and $\eps>0$ as input, runs in
time $\poly(n,2^{\tilde{O}(1/\eps^{9})})$
and outputs a value $v \in [0,1]$ such that
$\left|\Pr_{x \in \{-1,1\}^n}[q(x) \geq 0] - v \right| \leq \eps.$
\end{theorem}
\noindent Note that, as a consequence of this theorem, we get a $\poly(n)$ time deterministic algorithm to
count the fraction of satisfying assignments of a degree-$2$ PTF over $\{-1,1\}^n$ up to error $\eps = \tilde{O}(\log^{-1/9} n)$.

The \emph{influence} of variable $i$ on a polynomial $p: \{-1,1\}^n \to \R$,
denoted $\Inf_i(p)$,
is the sum of squares of all coefficients of $p$ that are on
monomials containing $x_i$; it is a measure of how much ``effect" the
variable $i$ has on the outcome of $p$.  Following previous
work \cite{DHK+:10} we say that a polynomial $p$ is
\emph{$\eps$-regular} if $\max_{i \in [n]}
\Inf_i(p) \leq \eps \cdot \sum_{j=1}^n \Inf_j(p).$
For sufficiently regular polynomials, our algorithm
runs in fully polynomial time $\poly(n,1/\eps)$:

\begin{theorem} \label{thm:count-regular-ptf-informal}
[Deterministic approximate counting of regular degree-2 PTFs over
$\{-1,1\}^n$, informal statement]
Given $\eps>0$ and an $O(\eps^{9})$-regular
degree-2 polynomial $q(x_1,\dots,x_n)$
our algorithm runs (deterministically) in
time $\poly(n,1/\eps)$
and outputs a value $v \in [0,1]$ such that
$
\left|\Pr_{x \in \{-1,1\}^n}[q(x) \geq 0] - v \right| \leq \eps.
$
\end{theorem}

\noindent
We note that the regular case has been a bottleneck in all known constructions of explicit PRGs
for PTFs; the seed-length of known generators for regular PTFs is no better
than for general PTFs.  Given Theorem \ref{thm:count-regular-ptf-informal}, the only
obstacle to improved running times for deterministic approximate counting algorithms
is improving the parameters of the ``regularity lemma'' which we use.
\footnote{Indeed, Kane \cite{Kane:13} has suggested that using the notions of regularity
and invariance from \cite{Kane12} may result in an improved, though still
$2^{\poly(1/\eps)}$, running time for our approach; we have not explored that in this work.}

\smallskip

\noindent
{\bf Discussion.}
Our counting results described above give deterministic
\emph{additive} approximations
to the desired probabilities.  While additive approximation
is not as strong as multiplicative
approximation, we recall that the problem of
determining whether $\Pr_{x \in \{-1,1\}^n}[q(x) \geq 0]$ is nonzero
is well-known to be NP-hard for degree-2 polynomials even if all nonconstant monomials in $q$ are
restricted to have coefficient 1 (this follows by a simple
reduction from Max-Cut, see the polynomial
$q_{G,\mathrm{CUT}}$ defined below).  Thus, no
efficient algorithm, even allowing randomness, can give any
multiplicative approximation to $\Pr_{x \sim \{-1,1\}^n}[q(x) \geq 0]$
unless $\mathbf{NP} \subseteq \mathbf{RP}$.  Given this, it is natural to consider
additive approximation.

Our results for degree-$2$ PTFs yield efficient deterministic
algorithms for a range of natural problems.  As a simple example,
consider the following problem: Given an undirected $n$-node
graph $G=([n],E)$ and a size parameter $k$, the goal is to estimate
the fraction of all $2^{n-1}$ cuts that contain at least $k$ edges.
(Recall that \emph{exactly} counting the number
of cuts of at least a given size is known to be \#P-hard
~\cite{Papadimitriou}.)
We remark that a simple sampling-based approach yields an efficient
\emph{randomized} $\pm \eps$-additive approximation algorithm for this problem.
Now note that the value of the polynomial $q_{G,\mathrm{CUT}}(x)
=(|E| - \sum_{\{i,j\} \in E} x_i x_j)/2$ on input $x \in \{-1,1\}^n$ equals
the number of edges in the cut corresponding to $x$ (where vertices
$i$ such that $x_i=1$ are on one side of the cut and vertices $i$ such that $x_i=-1$ are on the other side).  It is
easy to see that if $|E| \geq C^{9} n$ then $q_{G,\mathrm{CUT}}(x)$
must be $(1/C^{9})$-regular.
Theorem~\ref{thm:count-regular-ptf-informal} thus provides a \emph{deterministic} poly$(n,1/C)$-time algorithm that gives an
$\pm O(1/C)$-additive approximation to the fraction of all cuts that have
size at least $k$ in $n$-node graphs with at least $C^{9}n$ edges,
and Theorem~\ref{thm:count-boolean-ptf-informal} gives a deterministic
poly$(n,2^{\tilde{O}(1/\eps^{9})})$-time $\pm \eps$-approximation
algorithm for all $n$-node graphs.

As another example, consider the
polynomial $q_{G,\mathrm{INDUCED}}(x)
= \sum_{\{i,j\} \in E} {\frac {1 + x_i} 2} \cdot
{\frac {1 + x_j} 2}$.  In this case, we have that $q_{G,\mathrm{INDUCED}}(x)$ equals the number
of edges in the subgraph of $G$ that is induced by
vertex set $\{i: x_i=1\}$.
Similarly to the example of the previous paragraph,
Theorem~\ref{thm:count-regular-ptf-informal}
yields a deterministic poly$(n,1/C)$-time
algorithm that gives a $\pm O(1/C)$-additive approximation
to the fraction of all induced subgraphs
that have at least $k$ edges
in $n$-node graphs with at least $C^{9}n$ edges,
and Theorem~\ref{thm:count-boolean-ptf-informal} gives a deterministic
poly$(n,2^{\tilde{O}(1/\eps^{9})})$-time $\pm \eps$-additive approximation
algorithm for any graph.

\medskip

\noindent {\bf Estimating moments.}
Our results also imply deterministic fixed-parameter tractable
(FPT) algorithms for approximately computing moments of degree-$2$
polynomials.  Consider the following computational problem {\sc Absolute-Moment-of-Quadratic}:
given as input a degree-2 polynomial $q(x_1,\dots,x_n)$
and an integer parameter $k \geq 1$, output the value $\E_{x \in \{-1,1\}^n}[|q(x)|^k].$
It is clear that the \emph{raw} moment $\E[q(x)^k]$ can be computed
exactly in $n^{O(k)}$ time by expanding out the polynomial $q(x)^k$,
performing multilinear reduction, and outputting the constant term.
Since the $k$-th raw moment equals the $k$-th absolute moment when $k$ is
even, this gives an $n^{O(k)}$ time algorithm for
\noindent {\sc Absolute-Moment-of-Quadratic} for even $k$.  However,
for any fixed odd $k \geq 1$ the {\sc Absolute-Moment-of-Quadratic}
problem is \#P-hard
\ifnum\confversion=1
(see Appendix~B in the full version).
\else
(see Appendix~\ref{ap:moment-hard}).
\fi
Thus, it is natural to seek approximation algorithms
for this problem.

\ignore{
In contrast, we show in Appendix~\ref{ap:momenthard}
(building on recent work of Bl\"{a}ser and Curticapean
\cite{BC12}) that the {\sc Raw-Moment-of-Quadratic} and
{\sc Absolute-Moment-of-Quadratic} problems are both
\#$W[1]$-hard with respect to parameter $k$.
This implies that there is most likely
no $f(k) \cdot n^c$-time algorithm that solves these problems
for any computable function $f$ and fixed constant $c$
(otherwise there would be a $2^{o(n)}$-time algorithm for counting
the satisfying assignments of an $n$-variable 3-CNF formula;
see \cite{FlumGrohe04} for background on the counting
complexity class \#$W[1]$).  Thus it is natural to seek approximation
algorithms for these problems.
}

\ignore{
The hypercontractive inequality {\bf [[[xxxcitesxxx]]]} implies
that any degree-2 polynomial $q(x_1,\dots,x_n)$
with sum of squared coefficients 1 has $\E_{x \in \{-1,1\}^n}
[q(x)^k] \leq \new{k^{2k}}$, \new{(((how small can it be,
let's give a lower bound?)))} and thus it is meaningful to
seek additive $\pm \eps$-approximation algorithms for
{\sc Moment-of-Quadratic}.
}
Using the hyper-contractive inequality \cite{Bon70,Bec75}
it can be shown that the natural randomized algorithm
-- draw uniform points from $\{-1,1\}^n$ and use
them to empirically estimate $\E_{x \in \{-1,1\}^n}[|q(x)|^k]$ --
with high probability
gives a $(1 \pm \eps)$-accurate estimate of
the $k$-th absolute moment of $q$ in
$\poly(n,2^{k \log k}, 1/\eps)$ time.
Using Theorem~\ref{thm:count-boolean-ptf-informal} we are able to
derandomize this algorithm and obtain a
\emph{deterministic} FPT $(1 \pm \eps)$-multiplicative
approximation algorithm for {\sc Absolute-Moment-of-Quadratic}:

\begin{theorem} \label{thm:compute-kth-moment}
There is a {\em deterministic} algorithm which, given any degree-$2$
polynomial $q(x_1,\dots,x_n)$
with $b$-bit integer coefficients, any integer
$k \geq 1$, and any $\eps \in (0,1)$,
runs in
$\poly \left( n,b,2^{\tilde{O}((k \log k \log(1/\eps))^{9k}/\eps^9)}
\right)$ time
and outputs a value $
v \in \left[
(1-\eps) \E_{x \in \{-1,1\}^n}[|q(x)|^k],
(1+\eps) \E_{x \in \{-1,1\}^n}[|q(x)|^k]
\right]
$ that multiplicatively
$(1 \pm \eps)$-approximates the $k$-th absolute moment of $q$.
\end{theorem}

\subsection{Techniques.}

The major technical work in this paper goes into proving
Theorem~\ref{thm:count-regular-ptf-informal}.  Given
Theorem~\ref{thm:count-regular-ptf-informal}, we use a (deterministic)
algorithmic version of the ``regularity lemma for degree-$d$ PTFs''
from \cite{DSTW:10} to reduce the case of general degree-2 PTFs
to that of regular degree-2 PTFs.  (The regularity lemma that is
implicit in \cite{HKM:09} can also be used for this purpose.)

As is usual in this line of
work, we can use the invariance principle of Mossel \etal~\cite{MOO10} to
show that for an $O(\eps^9)$-regular degree-$2$ polynomial $p : \mathbb{R}^n \rightarrow
\mathbb{R}$, we have
$
\left|\Pr_{x \in \{-1,1\}^n} [p(x) \ge 0] -  \Pr_{x \in \mathcal{N}(0,1)^n}
[p(x) \ge 0] \right| \le \eps.
$
Thus, to prove Theorem~\ref{thm:count-regular-ptf-informal},
 we are left with the task of additively
estimating $\Pr_{x \in \mathcal{N}(0,1)^n} [p(x) \ge 0]$.

The first conceptual idea towards achieving the aforementioned task is this: Since Gaussian distributions are invariant under
rotations, computing the probability of interest $\Pr_{x \in \mathcal{N}(0,1)^n} [p(x) \ge 0]$ is equivalent to
computing $\Pr_{x \in \mathcal{N}(0,1)^n} [\tilde{p}(x) \ge 0]$ for a
``decoupled'' polynomial $\tilde{p}$. More precisely, there exists a
polynomial $\tilde{p} : \mathbb{R}^n \rightarrow \mathbb{R}$ of the
form $\tilde{p}(x) = \littlesum_{i=1}^n \lambda_i x_i^2 + \littlesum_{i=1}^n
\mu_i x_i +C$ such that the distributions of $p(x)$ and $\tilde{p}(x)$
(where $x \sim \mathcal{N}(0,1)^n$) are identical.  Indeed,
consider the symmetric matrix $A$ associated with the quadratic part
of $p(x)$ and let $Q^T \cdot A \cdot Q = \Lambda$ be the spectral
decomposition of $A$. It is easy to show that $\tilde{p}(x) =
p((Q \cdot x)_1, \ldots, (Q \cdot x)_n)$ is a decoupled polynomial with the same distribution as $p(x)$,  $x \sim \mathcal{N}(0,1)^n$.
The counting problem for $\tilde{p}$ should intuitively be significantly easier since there are no correlations between $\tilde{p}$'s monomials, and
hence it would be useful if $\tilde{p}$ could be efficiently exactly obtained from $p$.
Strictly speaking, this is not possible, as one cannot
obtain the exact spectral decomposition of a symmetric matrix $A$ (for example, $A$ can have irrational eigenvalues).
For the sake of this informal discussion, we assume that one can in fact obtain the exact decomposition and hence the polynomial $\tilde{p}(x)$.

Suppose we have obtained the decoupled polynomial $\tilde{p}(x)$. The second main idea in our approach is the following:
We show that one can efficiently construct a $t$-variable ``junta'' polynomial
$q : \mathbb{R}^t \rightarrow \mathbb{R}$, with $t = \poly(1/\eps)$,
such that the distribution of $q(x)$ is $O(\eps)$-close to the distribution of $\tilde{p}(x)$ in
Kolmogorov distance. (Recall that the Kolmogorov distance between two random variables is the maximum distance between their CDFs.)
To prove this, we use a powerful recent result of Chatterjee~\cite{Chatterjee:09}
\ifnum\confversion=0
(Theorem~\ref{thm:chat}),
\fi
\ifnum\confversion=1
(see Theorem~42 of the full version),
\fi
proved using Stein's method,
which provides a {\em central limit theorem} for functions of Gaussians.
Informally, this CLT says that for any function $F: \mathbb{R}^n \rightarrow \mathbb{R}$ satisfying some technical conditions,
if $g_1, \ldots, g_n$ are independent $\mathcal{N}(0,1)$ random variables,
then $F(g_1, \ldots, g_n)$ is close in total variation distance ($\ell_1$ distance between the pdfs)
to a Gaussian distribution with the ``right" mean and variance. {(We refrain from giving a more detailed description of the theorem here as the technical conditions  stem from considering generators of the Ornstein-Uhlenbeck process, thus rendering it somewhat unsuitable for an intuitive discussion.)} 
Using this result, we show that if $\max_i \lambda_i^2 \le \eps^2
\cdot \Var(\tilde{p})$ (i.e., if the maximum magnitude eigenvalue of the symmetric matrix $A$ corresponding to $p$ is ``small''), then the distribution of $\tilde{p}(x)$
(hence, also of $p(x)$) is $O(\eps)$-close to $\mathcal{N}(\mathbf{E}[\tilde{p}], \Var(\tilde{p}))$,
and hence the one-variable polynomial $q(x) = \sqrt{\Var(\tilde{p})} x_1 +
\mathbf{E}[\tilde{p}]$ is the desired junta.
In the other case, i.e., the case that
$\max_i \lambda_i^2 > \eps^2 \cdot \Var(\tilde{p})$, one must resort to
a more involved approach as described below.

If $\max_i \lambda_i^2 > \eps^2 \cdot \Var(\tilde{p})$,
we perform a ``critical index based'' case analysis
(in the style of Servedio, see \cite{Servedio:07cc}) appropriately tailored to the current setting.
We remark that such analyses have been used several times in the study of
linear and polynomial threshold functions (see e.g., \cite{DGJ+09, FGRW09, DHK+:10, DSTW:10}).
In all these previous works the critical index analysis was
performed on {\em influences} of variables in the original polynomial (or linear
form).  Two novel aspects of the analysis in the current work are that
(i) we must transform the polynomial from its original form into the
``decoupled'' version before carrying out the critical index analysis; and
(ii) in contrast to previous works, we perform the critical index analysis not on the influences of
variables, but rather on the {\em eigenvalues} of the quadratic part
of the decoupled polynomial,
i.e.,  on the values $(|\lambda_1|, \ldots, |\lambda_n|)$,
ignoring the linear part of the decoupled polynomial. The following paragraph explains
the situation in detail.

Suppose that the eigenvalues are ordered so that $|\lambda_1|
\ge \ldots \ge |\lambda_n|$. Consider the polynomials $\tilde{p}_{H,i} (x)=
C+ \littlesum_{j\le i} (\lambda_j x_j^2 + \mu_j x_j)$ (the ``head part'')
and  $\tilde{p}_{T,i}(x) = \littlesum_{j>i} (\lambda_j x_j^2 + \mu_j x_j)$
(the ``tail part'').
Define the $\tau$-critical index as the minimum $\ell \in [n]$
such that $|\lambda_{\ell}|/
\sqrt{\Var(\tilde{p}_{T,\ell-1})} \le \tau$. Let $K_0 = \Theta (\tau^{-2}\log (1/\tau))$.
If the $\tau$-critical index is more than $K_0$ then we show that the ``head part"  $\tilde{p}_{H,K_0} (x)$ (appropriately shifted)
captures the distribution of $\tilde{p}(x)$ up to a small error. In particular, the
distribution of $q(x) = \tilde{p}_{H,K_0} (x)+ \mathbf{E} [\tilde{p}_{T,K_0}(x)]$
is $O(\sqrt{\tau})$-close to that of $\tilde{p}$ in Kolmogorov distance.
On the other hand, if the critical index is $K \le K_0$, then it follows from Chatterjee's CLT that the polynomial
$q(x) = \tilde{p}_{H,K} (x)+ \sqrt{\Var(\tilde{p}_{T,K}(x))} x_{K+1} + \mathbf{E}[\tilde{p}_{T,K}(x)]$ is $O(\tau)$-close
to $\tilde{p}$ in total variation distance (hence, also in Kolmogorov distance). Note that in both cases, $q(x)$ has at most $K_0 +1$
variables and hence setting $\tau = \Theta (\eps^2)$, we obtain a polynomial $q(x)$
on $t  \le K_0+1 = \poly(1/\eps)$ variables whose distribution is Kolmogorov $O(\eps)$-close to
that of $\tilde{p}(x)$.

Thus, we have effectively reduced our initial task to the deterministic approximate computation of
$\Pr_{x \sim \mathcal{N}(0,1)^{K_0+1}} [q(x) \ge 0]$.
This task can potentially be achieved in a number of different ways
\ifnum\confversion=1
(see the discussion at the start of Section \ref{sec:second-stage}
of the full version);
\fi
\ifnum\confversion=0
(see the discussion at the start of Section \ref{sec:second-stage});
\fi
with the aim of giving
a self-contained and poly$(1/\eps)$-time algorithm, we take a
straightforward approach based on dynamic programming.
To do this, we start by discretizing the random variable $\mathcal{N}(0,1)$ to obtain a
distribution $D_{\mathcal{N}}$ (supported on $\poly(1/\eps)$ many points)
which is such that
$
\left| \Pr_{x \sim \mathcal{N}(0,1)^{K_0+1}} [q(x) \ge 0] -
\Pr_{x \sim D_{\mathcal{N}}^{K_0+1}} [q(x) \ge 0] \right| \le \eps.
$
Since $q(x)$ is a decoupled polynomial, computing $ \Pr_{x \sim
D_{\mathcal{N}}^{K_0+1}} [q(x) \ge 0]$ can be reduced to the counting
version of the knapsack problem where the weights are integers of magnitude $\poly(1/\eps)$,
 and therefore can be solved exactly in time $\poly(1/\eps)$ by standard dynamic programming.

\begin{remark}
We note that the dynamic programming approach we employ could be used to do
deterministic approximate counting for a decoupled $n$-variable Gaussian degree-2
polynomial $\tilde{p}(x)$ in $\poly(n,1/\eps)$ time even without the junta condition.  However,
the fact that $\tilde{p}$ is Kolmogorov-close to a junta polynomial
$q$ is a structural insight which has already proved useful in followup work.  Indeed,
achieving a junta structure is absolutely crucial for recent extensions of this
result \cite{DDS14junta,DS14degd} which generalize the deterministic approximate counting algorithm
presented here
(to juntas of degree-2 PTFs in \cite{DDS14junta} and to general degree-$d$ PTFs in
\cite{DS14degd}, respectively).
\end{remark}


\noindent \textbf{Singular Value Decomposition:}  The above informal
description glossed over the fact that given a matrix $A$, it is in
general not possible to exactly represent the SVD of $A$ using a
finite number of bits (let alone to exactly compute the SVD in
polynomial time). In our actual algorithm, we have to deal with
the fact that we can only achieve an ``approximate'' SVD.
We  define a notion of approximation that is sufficient for our purposes and
show that such an approximation can be computed efficiently.
Our basic algorithmic primitive is (a variant of the) well-known
``powering method'' (see~\cite{Vis:13} for a nice overview). Recall that the powering method efficiently computes an approximation to
the eigenvector corresponding to the highest magnitude eigenvalue.
In particular, the method has the following guarantee: given that the largest-magnitude
eigenvalue of $A$ has absolute value $|\lambda_{\max}(A)|$, the powering
method runs in time $\poly(n,1/\kappa)$ and returns a unit vector $w$ such that $\Vert A \cdot w \Vert_2 \ge |\lambda_{\max} (A)| \cdot (1-\kappa)$.

For our purposes, we require an additional criterion: namely,
that  the vector $A \cdot w$ is almost parallel to $w$. (This corresponds to the
notion of ``decoupling'' the polynomial discussed earlier.)
It can be shown that if one naively applies the powering method, then it is
not necessarily the case that the vector $w$ it returns will also satisfy
this requirement. To get around this, we modify the matrix $A$
before applying the powering method and show that the vector $w$ so
returned provably satisfies the required criterion, i.e.,~$A \cdot w$ is
almost parallel to $w$. An additional caveat is that the
standard ``textbook'' version of the method is a randomized
algorithm, and we of course need a deterministic algorithm.  This
can be handled by a straightforward derandomization,  resulting in only a linear time overhead.

\ifnum\confversion=0

\subsection{Organization.}

We record basic background facts from linear algebra,
probability, and analysis in Appendix~\ref{sec:background}, along
with our new extended notion of the ``critical index'' of a pair of
sequences.
Section~\ref{sec:deg2-gauss-count} establishes
our main technical result -- an
algorithm for deterministically approximately counting satisfying
assignments of a degree-2 PTF under the Gaussian $\mathcal{N}(0,1)^n$
distribution.
Section~\ref{sec:count-bool} extends this result to satisfying assignments
over $\{-1,1\}^n$.
Finally, in Section \ref{sec:moments}
we give the application to deterministic approximation of absolute moments.
\fi

\ifnum\confversion=1

{\bf Organization.}
Section~\ref{sec:deg2-gauss-count}
sketches
our main technical result -- an
algorithm for deterministically approximately counting satisfying
assignments of a degree-2 PTF under the Gaussian $\mathcal{N}(0,1)^n$
distribution.
Section~\ref{sec:count-bool} extends this result to satisfying assignments
over $\{-1,1\}^n$.
We give the application to deterministic approximation of absolute moments in Section~4
of the full version.
\fi

\section{Deterministic approximate counting for Gaussian distributions} \label{sec:deg2-gauss-count}

\subsection{Intuition.}

 Our goal is to compute,  up to an additive $\pm \eps$, the
probability $\Pr_{x \sim \mathcal{N}(0,1)^n}[p(x) \ge 0]$.  The algorithm has two
main stages.  In the first stage (Section~\ref{sec:first-stage}) we
transform the original $n$-variable degree-$2$ polynomial $p$ into
an essentially equivalent polynomial $q$ with a ``small'' number of
variables -- independent of $n$ --  and a
nice special form (a degree-$2$ polynomial with no ``cross terms").
The key algorithmic tool used in this transformation is the routine \textsf{APPROXIMATE-DECOMPOSE}  which is described
and analyzed in Section~\ref{sec:primitive}.
In particular, suppose that the original degree--$2$ polynomial is of the form $p(x) =
\littlesum_{i < j} a_{i,j} x_i x_j + \littlesum_{i} b_i x_i + C = x^T
A x +b^T x + C$.
The first stage constructs a degree-$2$ ``junta'' polynomial
$q(y_1,\dots,y_K):\R^K \to \R$ with no cross terms
(that is, every non-constant monomial in $q$ is either of the
form $y_i$ or $y_i^2$) where $K = \tilde{O}(1/\eps^4)$,
such that $| \Pr_{x \sim \mathcal{N}(0,1)^n}[p(x) \ge 0] -
\Pr_{y \sim \mathcal{N}(0,1)^K}[q(y) \ge 0] |  \le \eps.$
Theorem~\ref{thm:construct-junta-PTF}
summarizes what is accomplished in the first stage.
We view this stage as the main contribution of the paper.

In the second stage (Section~\ref{sec:second-stage})
we give an efficient deterministic algorithm to approximately count
the fraction of satisfying assignments for $q$.  Our algorithm exploits
both the fact that $q$ depends on only $\poly(1/\eps)$ variables
and the special form of $q$.
Theorem~\ref{thm:count-junta} summarizes what is accomplished in the second
stage.
Theorem~\ref{thm:count-gaussian-ptf} combines these two results and gives our
main result for deterministic approximate counting of Gaussian degree-2
PTFs.

\subsection*{The first stage:  Constructing a degree-2 junta PTF.}
To implement the first step we take advantage of the fact
that $x \sim \mathcal{N}(0,1)^n$ in order to ``decouple'' the variables.
Suppose we have computed the spectral decomposition of $A$ as $A = Q \Lambda Q^T$.
(We remark that our algorithm does not compute this decomposition explicitly; rather,
it iteratively approximates the eigenvector corresponding to the largest
magnitude eigenvalue
of $A$, as is described in detail in the pseudocode
of algorithm~{\tt Construct-Junta-PTF}.
For the sake of this intuitive explanation, we assume that we
construct the entire spectrum.)
Then, we can write $p$ as $$p(y) = y^T \Lambda y + \mu^T y + C = \littlesum_{i=1}^n \lambda_i y_i^2 +\littlesum_{i=1}^n \mu_i y_i
+ C,$$ where $y = Q^T x$ and $\mu = Q^T  b$. Since $Q$ is orthonormal, it follows that $y \sim \mathcal{N}(0,1)^n$ and that
the desired probability can be equivalently written
as $\Pr_{y \sim \mathcal{N}(0,1)^n}[p(y) \ge 0]$.
\ifnum\confversion=1
Before we proceed, we will need to define the precise notion of the critical index
that we use (which is a slightly augmented version of the definition used in the introduction).
\begin{definition} \label{def:ci}
Given a pair of sequences of non-negative numbers $\{c_i\}_{i=1}^n$
and $\{d_i\}_{i=1}^n$ where additionally the sequence $\{c_i\}_{i=1}^n$
is non-increasing, the \emph{$\tau$-critical index} of the pair
is defined to  be the smallest $0 \le i \le n-1$ such that
$
\frac{c_{i+1}}{\sum_{j>i}( c_j + d_j)} \le \tau.
$
In case there is no such number, we define the critical index to be $\infty$.
We call the sequence $\{c_i\}_{i=1}^n$  the ``main sequence" and  
$\{d_i\}_{i=1}^n$ the ``auxiliary sequence".

\end{definition}
\fi

At this point, let us arrange the variables in order, so that the sequence
$|\lambda_1|, \ldots, |\lambda_n|$ is non-increasing.
We now consider the $\eps$-critical index of the pair
of sequences $\{\lambda_i^2 \}_{i=1}^n$ and $\{\mu_i^2\}_{i=1}^n$
(here $\{ \mu_i^2 \}_{i=1}^n$ is the ``auxiliary sequence"\ifnum \confversion=0
see Definition~\ref{def:ci}\fi). The starting point of
our analysis is the following result.

\medskip

\noindent
\textbf{Informal theorem:}
 If the $\eps$-critical index  is zero, then the random variable $p(y)$, where $y \sim \mathcal{N}(0,1)^n$,
 is $O(\sqrt{\eps})$-close in total variation distance to $\mathcal{N}(\nu, \sigma^2)$ where
 $\nu = \mathbf{E}[p(y)]$ and $\sigma^2 = \Var[p(y)]$.

\medskip

As mentioned earlier, the proof of the above theorem uses a recent result of
Chatterjee~\cite{Chatterjee:09}
\ifnum\confversion=0
(Theorem~\ref{thm:chat})
\fi
which provides a central limit theorem for functions of Gaussians.
With this as starting point,
 we consider a case analysis depending on the value of the $\eps$-critical index of the
pair of sequences $\{\lambda_i^2 \}_{i=1}^n$ and $\{\mu_i^2\}_{i=1}^n$
($\{ \mu_i^2 \}_{i=1}^n$ is the auxiliary sequence). Let $K$ be the value of the the $\eps$-critical index of the pair.
If $K \le K_0 \eqdef \tilde{O}(1/\eps^2)$, then the
tail $p_{T, K}(y) =    \littlesum_{j>K} (\lambda_j y_j^2 + \mu_j y_j)$ can be replaced
by $\mathcal{N}(\nu_j, \sigma_j^2)$ where $\nu_j = \mathbf{E} [p_{T, K}(y)]$ and $\sigma^2 =
\Var[p_{T, K}(y)]$.
On the other hand, if $K \ge K_0$, then the distribution
of $q(y) = \littlesum_{j \le K_{0}} (\lambda_j y_j^2 + \mu_j y_j) +C$ differs
from the distribution of $p(y)$ by $O(\sqrt{\eps})$ in Kolmogorov distance. In either
case, we end up with a degree-$2$ polynomial on at most $K_0+1=\tilde{O}(1/\eps^2)$
variables whose distribution is $O(\sqrt{\eps})$ close to
the distribution of $p(y)$ in Kolmogorov distance.

The main difficulty in the real algorithm and analysis vis-a-vis the idealized version
described above is that computationally, it is not possible to compute
the exact spectral decomposition. Rather, what one can achieve is some
sort of an approximate decomposition (we are deliberately being vague here
about the exact nature of the approximation that can be achieved). Roughly speaking, at every stage of the algorithm
constructing $q$ several approximations are introduced and non-trivial technical work is required in bounding the corresponding error.
See Sections~\ref{sec:primitive}
and~\ref{sec:first-stage}
\ifnum \confversion=1
of the full version
\fi
for the detailed analysis.

\subsection*{The second stage:  Counting satisfying assignments of
degree-$2$ juntas over Gaussian variables.}  We are now left with the task
of (approximately) counting $\Pr[q(y) \ge 0]$.  To do this we
start by discretizing each normal random variable $y_i$ to a sufficiently fine
granularity -- it turns out that a grid of size $\poly(1/\eps)$ suffices. Let us denote by $\tilde{y}_i$ the discretized approximation to $y_i$.
We also round the coefficients of $q$ to a suitable $\poly(\eps)$ granularity and denote by ${q}'$ the rounded polynomial.
It can be shown that $q(y)$ and $q'(\tilde{y})$ are $\eps$-close in Kolmogorov distance.
Finally, this reduces computing $\Pr[q(y) \ge 0]$ to
computing $\Pr[q'(\tilde{y}) \ge 0]$. Since the terms in $q'$
are decoupled (i.e., there are no cross terms) and have small integer
coefficients, $q'(\tilde{y})$ can be expressed as a read-once branching program of size $\poly(1/\eps)$.
Using dynamic programming, one can efficiently compute the exact probability
$\Pr[q'(\tilde{y}) \ge 0]$
in time $\poly(1/\eps)$.  See Section~\ref{sec:second-stage}
\ifnum \confversion=1
of the full version
\fi
for the details.

We note that alternative algorithmic approaches could potentially
be used for this stage.  We chose our approach of discretizing and using dynamic programming because
we feel that it is intuitive and self-contained and because it easily gives a
$\poly(1/\eps)$-time algorithm for this stage.

\subsection{A useful algorithmic primitive.} \label{sec:primitive}
\ifnum\confversion=1
\noindent
In this section we state  the main
algorithmic primitive \textsf{APPROXIMATE-DECOMPOSE} used by our
procedure for constructing a degree-2 junta over Gaussian
variables.
\fi
\ifnum\confversion=0
\noindent
In this section we state and prove correctness of the main
algorithmic primitive \textsf{APPROXIMATE-DECOMPOSE} that our
procedure for constructing a degree-$2$ junta over Gaussian
variables will use.
\fi
This primitive partially ``decouples'' a given input degree-$2$ polynomial
by transforming the polynomial into an (essentially equivalent)
polynomial in which a new variable $y$ (intuitively corresponding
to the largest eigenvector of the input degree-2 polynomial's matrix)
essentially does not appear together with any other variables in any
monomials.

Theorem~\ref{thm:degree-decomp} gives a precise statement of
\textsf{APPROXIMATE-DECOMPOSE}'s performance.
\ifnum\confversion=0
The reader who is eager to see how \textsf{APPROXIMATE-DECOMPOSE}
is used may wish to proceed directly
from the statement of
Theorem~\ref{thm:degree-decomp} to Section~\ref{sec:first-stage}.

\fi
\ifnum\confversion=1
The procedure \textsf{APPROXIMATE-DECOMPOSE}
is presented and analyzed in detail in
Section~\ref{sec:primitive} of the full version.
\fi
We require the following definition to state Theorem~\ref{thm:degree-decomp}.
(Below a ``normalized linear form'' is an expression
$\sum_{i=1}^n w_i x_i$ with $\sum_{i=1}^n w_i^2 = 1.$)

\begin{definition}\label{def:residue}
Given a degree-2 polynomial $p : \mathbb{R}^n \rightarrow \mathbb{R}$
defined by $p(x)  = \littlesum_{1 \le i \le j \le n} a_{ij} x_i x_j
+ \sum_{1 \le i \le n} b_i x_i +C$ and a normalized linear form
 $L_1(x)$,
we define the
\emph{residue of $p$ with respect to $L_1(x)$},
$\mathop{Res}(p, L_1(x))$, to be the polynomial obtained by the
following process : For each $i \in [n]$,
express $x_i$ as $\alpha_{i1} L_1(x) + R_i(x)$
where $R_i(x)$ is orthogonal to the linear form $L_1(x)$. Now, consider
the polynomial $q(y_1,x) = p(\alpha_{11} y_1 + R_1(x), \ldots,
\alpha_{n1}y_1 +R_n(x))$. $\mathop{Res}(p, L_1(x))$ is defined as the
homogenous multilinear degree-2
part of $q(y_1,x)$ which has the variable $y_1$
present in all its terms.
\end{definition}

\begin{theorem}\label{thm:degree-decomp}
Let $p : \mathbb{R}^n \rightarrow \mathbb{R}$ be a degree-$2$ polynomial
(with constant term $0$) whose entries are $b$-bit integers and
let $\epsilon, \eta>0$.  There exists a deterministic algorithm
\textsf{APPROXIMATE-DECOMPOSE} which on input an explicit description
of $p$, $\eps$ and $\eta$ runs in time $\poly(n,b,1/\eps,1/\eta)$ and
has the following guarantee :

\begin{itemize}

\item[(a)] If $\lambda_{\max}(p) \ge \eps \sqrt{\Var(p)}$, then
the algorithm outputs rational numbers $\lambda_1$, $\mu_1$ and a
degree-$2$ polynomial $r : \mathbb{R}^{n+1} \rightarrow \mathbb{R}$ with the
following property: for
$(y,x_1,\dots,x_n) \sim \mathcal{N}(0,1)^{n+1}$, the
   two distributions $p(x_1,\dots,x_n)$ and $q(y_1,x_1,\dots,x_n)$ are
   identical, where $q(y_1,x_1,\dots,x_n)$ equals $\lambda_1 y_1^2 + \mu_1
   y_1 + r(y_1,x_1,\dots,x_n).$
Further, $\Var(\mathop{Res}(r,y_1)) \le 4\eta^2 \Var(p)$ and $\Var(r)
\le (1-\eps^4/40) \cdot \Var(p)$.

\item[(b)] If $\lambda_{\max}(p) < \eps \sqrt{\Var(p)}$, then the algorithm
either outputs ``small max eigenvalue" or has the same guarantee as (a).

\end{itemize}
\end{theorem}

\ifnum\confversion=0
In the rest of Section~\ref{sec:primitive} we prove Theorem~\ref{thm:degree-decomp}, but first we give
some high-level intuition.
\fi
Recall from the introduction that we would like to
compute the SVD of the symmetric matrix corresponding to the quadratic
part of the degree-2 polynomial $p$, but the exact SVD is hard to compute.
\textsf{APPROXIMATE-DECOMPOSE} works by computing an approximation to the
largest eigenvalue-eigenvector pair, and using the approximate eigenvector to
play the role of $L_1$ in Definition \ref{def:residue}.

The case that is of most interest to us is when the largest eigenvalue
has large magnitude compared to the square root of the variance of $p$
(since we will use Chatterjee's theorem to deal with
the complementary case) so we focus on this case below.
For this case, part (a) of Theorem \ref{thm:degree-decomp}
says that the algorithm outputs a degree-$2$ polynomial $q(y_1,x_1,\dots,x_n)$
with the same distribution as $p$.  Crucially, in this polynomial $q$,
the first variable $y_1$ is ``approximately decoupled" from the rest of
the polynomial, namely $r$ (because $\Var(\mathop{Res}(r,y_1))$ is small),
and moreover $\Var(r)$ is substantially smaller than $\Var(p)$ (this
is important because intuitively it means we have ``made progress'' on
the polynomial $p$).
Note that if we were given the exact eigenvalue-eigenvector pair
corresponding to the largest magnitude eigenvalue, it would be possible to
meet the conditions of case (a) with $\eta=0$.

While approximating the largest eigenvector is a well-studied problem,
we could not find any off-the-shelf solution
with the guarantees we required.
\textsf{APPROXIMATE-DECOMPOSE} adapts the
well-known powering method for finding the largest eigenvector to give
the desired guarantees.
\ifnum\confversion=1
The details can be found in the full version of the paper.
\fi

\ifnum\confversion=0

\subsubsection{Decomposing a matrix.}

In order to describe the \textsf{APPROXIMATE-DECOMPOSE} algorithm
we first need a more basic procedure which we call
\textsf{APPROXIMATE-LARGEST-EIGEN}.
Roughly speaking, given a real symmetric matrix $A$ with a large-magnitude
eigenvalue, the \textsf{APPROXIMATE-LARGEST-EIGEN} procedure
outputs approximations of the largest-magnitude
eigenvalue and its corresponding eigenvector.
Theorem~\ref{thm:primitive} gives a precise performance guarantee:

\begin{theorem} \label{thm:primitive}
Let $A \in \R^{n \times n}$ be a symmetric matrix whose entries are $b$-bit
integers (not all 0) and $\eps,\eta>0$ be  given
rational numbers. There exists a deterministic algorithm
\textsf{APPROXIMATE-LARGEST-EIGEN} which on input $A$, $\eps$ and $\eta$, runs in time $\poly(n, b, 1/\eps,1/\eta)$
and has the following behavior:
\begin{itemize}
\item[(a)] If $|\lambda_{\max}(A)| \ge \eps \|A\|_F$,
the algorithm outputs
a number $\tilde{\lambda} \in \R_+$ and unit vector $\tilde{w} \in \R^n$
such that

\begin{itemize}
\item [(i)] $(1-\eta)
|\lambda_{\max}(A)| \le |\tilde{\lambda}
| \le |\lambda_{\max}(A)|$;
\item [(ii)] the matrix $\tilde{B} = A - \tilde{\lambda} (\tilde{w}\tilde{w}^{T})$ satisfies $\|\tilde{B}\tilde{w}\|_2 < \eta \|A\|_F$; and
\item [(iii)] $\| \tilde{B} \|_F \le (1-\eps^2/40) \cdot \Vert A \Vert_F$.
\end{itemize}

\item[(b)] If $|\lambda_{\max}(A)| < \eps \|A\|_F$,
the algorithm either outputs ``small max eigenvalue'' or behaves
as in case (a).
\end{itemize}

\ignore{
\inote{Let's see exactly what we need from our algorithm. Do we need a separation really?}
}

\end{theorem}

Let us describe the \textsf{APPROXIMATE-LARGEST-EIGEN} algorithm.
Let $2^{-m} \le \epsilon \le 2^{-m+1}$. The running time of the
algorithm will have a polynomial dependence on $2^m$ .
Without loss of generality, assume that $\lambda_{\max}(A)$ is a positive
number. Instead of working directly with the matrix $A$, we will work with the
matrix $A ' = A + t \cdot I$ where $t = \lceil \Vert A \Vert_F \rceil$.
Note that an eigevector-eigenvalue pair $(v,\lambda)$ of $A$ maps to the
pair $(v,\lambda +t)$ for $A'$.

For $\delta = \min \{\eps^4/100, \eta^4/10^8 \}$,
the \textsf{APPROXIMATE-LARGEST-EIGEN} algorithm works as follows :

\begin{itemize}

\item For unit vectors $e_1, \ldots, e_n$
and $k =  \lceil \frac{1}{2\delta} \cdot \log (9n/4) \rceil$,
the algorithm computes $$\mu_i = \frac{\Vert A' \cdot (A'^k \cdot e_i)
\Vert_2^2}{ \Vert (A'^k \cdot e_i) \Vert_2^2}.$$

\item Let $i^{\ast} = \arg\max_{i \in [n]} \mu_i$, and define
  $$
w = \frac{A'^k \cdot e_{i^{\ast}}}{\Vert A'^k \cdot e_{i^{\ast}} \Vert_2},
\quad \lambda = \Vert A \cdot w \Vert_2.
$$
Note that since $w$ can have irrational entries,
exact computation of $w$ and $\lambda$ is not feasible. However, in time $\poly(1/\delta, b,n)$, we can compute a unit vector $\tilde{w}$ so that $\Vert w - \tilde{w} \Vert_2 \le \poly(\delta, 1/b,1/n)$.
Define $\tilde{\lambda}$ as $\Vert A \cdot \tilde{w} \Vert_2$
rounded to a precision $\poly(\delta,1/b,1/n)$. It is easy
to see that $| \tilde{\lambda} - \lambda| \le \poly(\delta,1/b,1/n)$.
\item If $\tilde{\lambda}^2 \ge (1-9 \cdot \delta^{1/4}) \cdot \eps^2 \cdot  \Vert A \Vert_F^2$, then output the pair $(\tilde{w},\tilde{\lambda})$. Else, output ``small max eigenvalue".
   \end{itemize}

\medskip
\noindent {\bf Proof of Theorem~\ref{thm:primitive}:}
We start with the following simple claim:

\begin{claim}\label{clm:eigen-1}
If $|\lambda_{\max} (A)| \le (\epsilon/2) \cdot \Vert A \Vert_F$,
then \textsf{APPROXIMATE-LARGEST-EIGEN} outputs ``small max eigenvalue".
\end{claim}
\begin{proof}
Note that if $|\lambda_{\max}(A)
|\le (\epsilon/2) \cdot \Vert A \Vert_F$, then
$\tilde{\lambda}^2 \le \epsilon^2 \Vert A \Vert_F^2/4$.
By our choice of $\delta$ we have that
$\tilde{\lambda}^2 <  (1-9 \cdot \delta^{1/4}) \cdot \eps^2 \cdot
\Vert A \Vert_F^2$, hence the algorithm will output ``small
max eigenvalue".
\end{proof}

Next let us recall the ``powering method" to compute the largest eigenvalue
of a symmetric matrix. See the monograph by Vishnoi~\cite{Vis:13}
(the following statement is implicit in Lemma~8.1).
\begin{lemma}\label{lem:powering}
Let $A \in \mathbb{R}^{n \times n}$ be a symmetric matrix, $\lambda_{\max}(A)$ be the largest magnitude eigenvalue of $A$ and $v$ be the corresponding eigenvector. Let $w$ be any unit vector such that $|\braket{v,w} | \ge \frac{2}{3 \sqrt{n}}$. Then, for $k  > \frac{1}{2\kappa} \cdot \log (9n/4)$, $\Vert A \cdot (A^k \cdot v) \Vert_2 \ge |\lambda_{\max} (A)| \cdot (1-\kappa) \cdot \Vert (A^k \cdot v) \Vert_2$.
\end{lemma}

Let $v_{\max}$ be the eigenvector corresponding to the largest
eigenvalue of $A'$ and $\lambda_{\max}(A')$ be the corresponding
eigenvalue. It is clear that there is some $i \in [n]$ such that
$|\braket{v_{\max},e_i}|
\ge \frac{1}{\sqrt{n}}$.

Let $i^{\ast}$ be any such index.
We will show that
\begin{equation} \label{eq:wlambda}
w = \frac{A'^k \cdot e_{i^{\ast}}}{\Vert A'^k \cdot e_{i^{\ast}} \Vert_2}
\quad \textrm{ and } \quad \lambda = \Vert A \cdot w \Vert_2
\end{equation}
are such that $\tilde{w}$ and $\tilde{\lambda}$ satisfy the conditions
given in (a) and (b) of Theorem~\ref{thm:primitive}.
Lemma~\ref{lem:powering} gives that $\sqrt{\mu_{i^\ast}} \geq (1-\delta) \cdot \lambda_{\max}(A')$, and hence $\|A' \cdot w\|_2^2 \geq (1-\delta)^2 \lambda_{\max}(A')^2$.

Let $v_1, \ldots, v_n$ be the eigenvectors of $A'$ (and hence of $A$)
and $\lambda'_1, \ldots, \lambda'_n$ be the corresponding eigenvalues of $A'$.
Let $w = \sum_{i=1}^n c_i \cdot v_i$ and let $S = \{i \in [n] :
\lambda'_i \ge (1-\sqrt{\delta}) \cdot \lambda_{\max} (A')\}$.

\begin{proposition}\label{prop:eigen-max1}
$\sum_{i \not \in S} c_i^2 \le 2 \sqrt{\delta}$.
\end{proposition}
\begin{proof}
We have $A' \cdot w = \sum_{i=1}^n c_i \cdot \lambda'_i \cdot v_i$ and hence
$ \Vert A' \cdot w \Vert_2^2 = \sum_{i=1}^n c_i^2 \cdot \lambda_i^{'2} =
\sum_{i \in S} c_i^2\cdot \lambda_i^{'2} + \sum_{i \not \in S} c_i^2 \cdot
\lambda_i^{'2}$.
As all eigenvalues of $A'$ are non-negative,
for $i \notin S$ we have $\lambda_i^{'2} \le  (1-\sqrt{ \delta})^2
\cdot \lambda_{\max}^2(A')$. If $\sum_{i \not \in S} c_i^2 = \kappa$, then
$$(1-\delta)^2 \lambda_{\max}^2(A') \le
\sum_{i \in S} c_i^2\cdot \lambda_i^{'2} + \sum_{i \not \in S} c_i^2 \cdot
\lambda_i^{'2} \le (1-\kappa) \lambda_{\max}^2(A') + \kappa (1-\sqrt{\delta})^2
\lambda_{\max}^2(A').$$
The last inequality uses that $\|A' \cdot w\|_2^2 \geq (1-\delta)^2 \lambda_{\max}(A)^2$. Thus, $-\kappa \sqrt{\delta} (2- \sqrt{\delta}) \ge -\delta (2-\delta)$ and hence $\kappa \le 2 \sqrt{\delta}$.
\end{proof}


 \begin{proposition}\label{prop:eigen-max2}
If $\lambda_{\max}(A) \ge \epsilon \cdot \Vert A \Vert_F$, then for $w$ as defined above, $\Vert A \cdot w \Vert_2^2  \ge(1-6 \cdot\delta^{1/4})  \cdot \lambda_{\max}^2(A) $.
\end{proposition}
\begin{proof}
Recall that an eigenvector $v_i$ with eigenvalue $\lambda_i$  of $A$ maps to an eigenvalue $\lambda_i +t $ of $A'$. Thus, if $i$ is such that $\lambda_i +t \ge (1-\sqrt{\delta}) (\lambda_{\max} (A) +t)$, then $\lambda_i \ge  (1 -\sqrt{\delta}) \lambda_{\max} (A) - \sqrt{\delta} \cdot t  $. Since $\lambda_{\max}(A) \ge  \epsilon t /2$, if we choose $\delta \le \epsilon^4$, then $\lambda_i \ge (1-2 \cdot \delta^{1/4}) \cdot \lambda_{\max}(A)$. Thus, we have
$$
S \subseteq \{ i \in [n] : \lambda_i \ge (1-2\cdot \delta^{1/4}) \cdot \lambda_{\max}(A)\}.
$$
Now, observe that $A \cdot  w = \sum_{i \in S} c_i \cdot \lambda_i \cdot v_i + \sum_{ i \not \in S} c_i \cdot \lambda_i \cdot v_i$. Hence,
$$
\Vert A \cdot w \Vert_2^2  \ge \sum_{i \in S} c_i^2 \cdot \lambda_i^2
\ge (1-2\cdot \delta^{1/4})^2 \cdot \lambda_{\max}^2(A)
\cdot (1 -2 \sqrt{\delta}) \ge (1-6 \cdot\delta^{1/4})  \cdot \lambda_{\max}^2(A),
$$
where the second inequality uses Proposition~\ref{prop:eigen-max1}.
\end{proof}

\begin{proposition}\label{prop:eigen-max3}
For  $w$ as defined in Equation~(\ref{eq:wlambda})
and $\lambda \ge 0$,  if $\lambda_{\max}(A) \ge (\eps/2) \cdot \Vert A \Vert_F$ and $\Vert A \cdot w \Vert_2^2 = \lambda^2  \ge (1-10 \cdot \delta^{1/4}) \cdot \lambda_{\max}^2(A) $, then for  $B = A - \lambda w \cdot w^{T}$, $\Vert B \cdot w \Vert_2 \le \eta \cdot \Vert A \Vert_F/2$.
\end{proposition}
\begin{proof}
We begin by noting that for $i \in S$, $\lambda'_i \ge (1-\sqrt{\delta}) \lambda_{\max}(A')$. This implies that
$
\lambda_i +t \ge (1-\sqrt{\delta}) (\lambda_{\max}(A) +t)
$.
Using the bounds $\lambda_{\max}(A) \ge (\eps/2) \Vert A \Vert_F$ and
$\delta \le \eps^4$,  we get that $\lambda_i \ge (1-2 \cdot \delta^{1/4})
\cdot \lambda_{\max}(A)$.

By assumption we have that $(1-10 \cdot  \delta^{1/4}) \lambda_{\max}(A) \le
\lambda$, and since $\|A \cdot w \|_2^2 = \lambda^2$ we also have that
$\lambda \le \lambda_{\max}(A)$.  As a consequence, we have that for every $i \in S$,  $|\lambda- \lambda_i| \le 10 \cdot \delta^{1/4} |\lambda_{\max}(A)|$.
Note that
\begin{eqnarray*}
\Vert B \cdot w \Vert_2^2= \Vert A \cdot w - \lambda \cdot w \Vert_2^2 &=& \sum_{i \in S} c_i^2 (\lambda_i - \lambda)^2 + \sum_{i \not \in S} c_i^2 (\lambda_i - \lambda)^2 \\
& \le & 100 \cdot \sqrt{\delta}\cdot \lambda_{\max}^2(A) \cdot \left(\sum_{i \in S} c_i^2\right)  + \sum_{i \not \in S} 2 \cdot c_i^2 \cdot  (\lambda_i^2 + \lambda^2) \\
&\le& 100 \cdot \sqrt{\delta} \cdot \lambda_{\max}^2(A) + 8 \sqrt{\delta}  \Vert A \Vert_F^2 \le 108 \sqrt{\delta} \Vert A \Vert_F^2 \le \eta^2 \cdot \Vert A
\Vert_F^2 /4,
\end{eqnarray*}
where we used Proposition~\ref{prop:eigen-max1} and Fact~\ref{fact:lin-alg}
in the last line.
The last inequality holds because $\delta \le \eta^4/10^8$.
\end{proof}

\begin{claim}\label{clm:eigen-max2}
If $\lambda_{\max}(A) \ge \epsilon \cdot \Vert A \Vert_F$, then the
output satisfies the guarantees of part (a)
of Theorem~\ref{thm:primitive}.
\end{claim}

\begin{proof}
Recall that $\delta = \min \{\eps^4/100, \eta^4/10^8 \}$.
We can then combine Proposition~\ref{prop:eigen-max1}, Proposition~\ref{prop:eigen-max2} and Proposition~\ref{prop:eigen-max3} to get that $(1 -\eta/2) \lambda_{\max}(A) \le \lambda \le \lambda_{\max}(A)$ and $\Vert B w \Vert_2 \le (\eta/2) \cdot \Vert A \Vert_F$. Finally, we use
Lemma~\ref{lem:variance-gets-smaller} (proved below)
to get that $\Vert B \Vert_F \le (1-\eps^2/20) \Vert A \Vert_F$.

Now, recall that $\Vert w - \tilde{w} \Vert_2 \le \poly(\delta, 1/b , 1/n)$ and $| \lambda- \tilde{\lambda}| \le \poly(\delta,1/b,1/n)$. This implies guarantees (i), (ii) and (iii).
\end{proof}


\begin{claim}\label{clm:eigen-max3}
If $\epsilon \ge \lambda_{\max}(A) \ge \epsilon /2$, then the output
satisfies the conditions in part (b) of Theorem~\ref{thm:primitive}.
\end{claim}

\begin{proof}
If  $\tilde{\lambda}^2 \le (1-9 \cdot \delta^{1/4}) \cdot \eps^2 \cdot
\Vert A \Vert_F^2$, then the algorithm outputs ``small max eigenvalue"
and the output is correct.
On the other hand, if $\tilde{\lambda}^2 \ge (1-9 \cdot \delta^{1/4})
\cdot \eps^2 \cdot  \Vert A \Vert_F^2$, then it implies that $\lambda^2
\ge (1- 10 \cdot \delta^{1/4}) \cdot \eps^2 \cdot \Vert A \Vert_F^2$.
As $\lambda_{\max}^2(A) \le \eps^2 \Vert A \Vert_F^2$,
hence $\lambda^2 \ge (1- 10 \cdot \delta^{1/4}) \cdot
\lambda_{\max}^2$.  Since $\delta \le \eta^4/10^8$, we get
that  $(1 -\eta/2) \lambda_{\max}(A) \le \lambda \le \lambda_{\max}(A)$.
As above, Proposition~\ref{prop:eigen-max1}, Proposition~\ref{prop:eigen-max3}
and Proposition~\ref{prop:eigen-max3}
give that $\Vert B w \Vert_2 \le (\eta/2) \cdot \Vert A \Vert_F$, and
Lemma~\ref{lem:variance-gets-smaller} gives that $\Vert B \Vert_F
\le (1-\eps^2/20) \Vert A \Vert_F$.
As before, using that $\Vert w - \tilde{w} \Vert_2 \le \poly(\delta, 1/b , 1/n)$
and $| \lambda- \tilde{\lambda}| \le \poly(\delta,1/b,1/n)$, we get guarantees
(i), (ii) and (iii) from the output.
\end{proof}


Claims~\ref{clm:eigen-1},~\ref{clm:eigen-max2}
and~\ref{clm:eigen-max3} together establish Theorem~\ref{thm:primitive}
modulo the proof of
Lemma~\ref{lem:variance-gets-smaller}, which we provide below.
\qed

\begin{lemma} \label{lem:variance-gets-smaller}
Let $A \in \R^{n \times n}$ be symmetric with $\|A\|_F=1$. For $0<\delta \le \eps<1$,
let $\lambda$ with $|\lambda| \ge 3\eps$ and
$v \in \R^n$ with $\|v\|_2=1$ be such that the matrix $B = A - \lambda (vv^T)$ satisfies $\|Bv\|_2 \le \delta$. Then, $\|B\|^2_F \le (1-3\eps^2).$
\end{lemma}

\begin{proof}
Recall that for any symmetric matrix $C \in \R^{n \times n}$ we have $\|C\|_F^2 = \tr(C^2).$ Hence we may prove the lemma by bounding
from above the quantity $\tr(B^2)$.
We can write
\[ B^2 = (A - \lambda vv^T)^2 = A^2 + \lambda^2 (vv^T)^2 - \lambda A (vv^T) - \lambda (vv^T)A\]
 and therefore
\[ \tr(B^2)  = \tr(A^2) + \lambda^2 \tr\left( (vv^T)^2\right) - \lambda \tr \left( A (vv^T) \right) - \lambda \tr \left( (vv^T)A \right).\]
Since $\|v\|_2 = 1$, we have that $ \tr\left( (vv^T)^2\right)=1.$ Moreover, $ \tr \left( A (vv^T) \right)  =  \tr \left( (vv^T)A \right).$
Therefore,
\[ \tr(B^2)  = 1 + \lambda^2 - 2 \lambda \tr \left( A (vv^T) \right).\]
We will need the following claim:
\begin{claim} \label{claim:trace}
We have that $|\tr (Bvv^T) | \le \delta.$
\end{claim}
\begin{proof}
It follows easily from the definition that $$\tr (Bvv^T) = \littlesum_{i=1}^n (b^{(i)} \cdot v)v_i,$$ where $b^{(i)} \in \R^n$ is the $i$-th row of $B$.
Since $Bv = [(b^{(1)} \cdot v) \ldots (b^{(n)} \cdot v)]^T$ the Cauchy-Schwarz inequality implies that
\[ |\tr (Bvv^T)| \le \|Bv\|_2 \| v\|_2.\]
The claim follows from the fact that $\| Bv\|_2 \le \delta$ and $\|v\|_2=1$.
\end{proof}
\noindent Since $$Avv^T = Bvv^T + \lambda (vv^T)^2$$ we have
\[  \tr(Avv^T) = \tr(Bvv^T) + \lambda \tr\left( (vv^T)^2 \right) = \tr(Bvv^T) + \lambda.\]
Claim~\ref{claim:trace} now implies that
\[
\lambda- \delta \le \tr(Avv^T) \le \lambda+\delta.
\]
Since $|\lambda| \ge\eps > \delta$, the numbers $\lambda-\delta, \lambda, \lambda+\delta$ have the same sign, hence
\[ \lambda \tr(Avv^T) = |\lambda| \cdot | \tr(Avv^T) |  \ge |\lambda| \cdot (|\lambda| - \delta) = \lambda^2 - \delta |\lambda|\]
which gives that
\[ \tr(B^2)  \le 1 + \lambda^2 - 2\lambda^2 + 2\delta|\lambda| = 1-\lambda^2+2\delta|\lambda| = 1-|\lambda| (|\lambda| - 2\delta) \le 1-3\eps^2\]
where the last inequality follows from our assumptions on $\lambda$ and $\delta$.
\end{proof}

\subsubsection{Technical claims about degree-$2$ polynomials.} \label{sec:tech}

Before presenting and analyzing the \textsf{APPROXIMATE-DECOMPOSE}
algorithm we need some technical setup.
We start with a useful definition:

\ignore{
We will now apply the above result to derive a useful claim about
decomposition of polynomials.}

\begin{definition}[Rotational invariance of polynomials]
Given two polynomials $p(x) = \littlesum_{1 \le i \le j \le n} a_{ij}x_i x_j
+ \littlesum_{1\le i \le n} b_i x_i +C$ and $q(x) =\littlesum_{1 \le i
\le j \le n} a'_{ij}x_i x_j + \littlesum_{1\le i \le n} b'_i x_i +C$
with the same constant term, we say that they are
\emph{rotationally equivalent} if there is an orthogonal matrix $Q$ such
that $Q^T \cdot A \cdot Q = A'$ and $ Q^T \cdot b = b'$.
If the matrix $A'$ is diagonal then the polynomial $q$ is said to be the
\emph{decoupled equivalent of $p$}.
In this case, the eigenvalues of $A$ (or
equivalently $A'$) are said to be the eigenvalues of the quadratic form $p$.
\end{definition}

\begin{claim}\label{clm:poly-equivalence}
For any degree-$2$ polynomials $p(x)$ and $q(x)$ which are rotationally equivalent, the distributions of $p(x)$ and $q(x)$ are identical when $(x_1, \ldots, x_n) \sim \mathcal{N}(0,1)^n$.
\end{claim}
\begin{proof}
Observe that $q(x) = p(y)$ where $y = Q \cdot x$.
Now, since $(x_1, \ldots, x_n)$ is distributed according to
$\mathcal{N}(0,1)^n$ and $Q$ is an orthogonal matrix,
$(y_1, \ldots, y_n)$ has the same distribution. This proves the claim.
\end{proof}
We will also use the following simple fact which relates the
variance of a degree-$2$ polynomial $p$ with the Frobenius norm of the
quadratic part of $p$.
\begin{fact}\label{fac:quad-norm}
Let $p : \mathbb{R}^n \rightarrow \mathbb{R}$ be a degree-$2$ polynomial and let $A$ be the matrix corresponding to its quadratic part. Then
 $\Var(p) \ge 2 \Vert A \Vert_F^2$.
\end{fact}

\begin{proof}
Let $p(x) = \littlesum_{1 \le i \le j \le n} a_{ij} x_i x_j +
\littlesum_{1 \le i \le n} b_i x_i +C$. Equivalently, $p(x) = x^T
\cdot A \cdot x + b^T \cdot x +C$ where $A$ is the matrix corresponding to
the quadratic part of $p$. Using the fact that $A$ is symmetric, we get
that there is an orthogonal matrix $Q$ such that $Q^T A Q = \Lambda$
where $\Lambda$ is diagonal. Using the fact that if $x \sim
\mathcal{N}(0,1)^n$, then so is $Q x$, we get that the distribution
of $p(x)$ and $q(x) = x^T Q^T A Q x + b^T  Qx +C$ are identical
when $x \sim \mathcal{N}(0,1)^n$.  However, $q(x) = x^T \Lambda x +
\mu^T x + C$ where $\mu = Q^T \cdot b$. Note that $q(x) =
\littlesum_{i=1}^n (\lambda_i x_i^2 + \mu_i x_i) +C$. Hence,
$$
\Var(q) = \sum_{i=1}^n \Var(\lambda_i x_i^2 + \mu_i x_i) = \sum_{i=1}^n 2
\lambda_i^2 + \mu_i^2 \ge 2 \sum_{i=1}^n \lambda_i^2 = 2 \Vert A \Vert_F^2
$$
(recall that for a univariate Gaussian random variable $x \sim \mathcal{N}(0,1)$
we have $\E[x^4]=3$ and hence $\Var(x^2) = 2$.).
Since the distributions of $p(x)$ and $q(x)$ are identical, we have that
$\Var(p) \ge 2  \Vert A \Vert_F^2$.

\end{proof}
We will also require another definition for degree-$2$ polynomials.
\begin{definition}
Given $p : \mathbb{R}^n \rightarrow \mathbb{R}$ defined by
$p(x) =
\littlesum_{1\le i \le j \le n} a_{ij} x_i x_j + \littlesum_{1 \le i \le n}
b_i x_i + C$,
define $\mathop{SS}(p)$ as $\mathop{SS}(p)=
\littlesum_{1\le i \le j \le n} a_{ij}^2 + \littlesum_{1 \le i \le n} b_i^2$.
\end{definition}

We now have the following simple claim.
\begin{claim}\label{clm:variance-bound-SS}
Given $p : \mathbb{R}^n \rightarrow \mathbb{R}$,
we have that $2 \mathop{SS}(p) \ge \Var(p) \ge  SS(p)$.
\end{claim}
\begin{proof}

Since neither $\mathop{SS}(p)$ nor $\Var(p)$ is changed by adding a constant
term to $p$ it suffices to prove the claim for
$p(x) =
\littlesum_{1\le i \le j \le n} a_{ij} x_i x_j + \littlesum_{1 \le i \le n}
b_i x_i.$ We have
\begin{eqnarray*}
\mathbf{E}_{x_1, \ldots, x_n  \sim \mathcal{N}(0,1)} [p(x_1, \ldots, x_n)^2]
&=& \sum_{i=1}^n b_i^2 \mathbf{E} [x_i^2] + \sum_{1 \le i < j \le n}
a_{ij}^2 \mathbf{E} [x_i^2 x_j^2] + \sum_{i=1}^n a_{ii}^2 \mathbf{E}[x_i^4]
+ \sum_{1 \le i <j \le n}2 a_{ii} a_{jj} \mathbf{E} [x_i^2 x_j^2] \\
&=& \sum_{i=1}^n (b_i^2 + 3 a_{ii}^2)
+ \sum_{1 \le i < j \le n}
(a_{ij}^2 + 2 a_{ii} a_{jj}).
\end{eqnarray*}
The first equality holds because every other cross-term has an
odd power of some $x_i$ (for $i \in [n]$),
and for $x \sim \mathcal{N}(0,1)$ we have that $\mathbf{E} [x^t]=0$
if $t$ is odd. On the other hand, by linearity of expectation,
$\mathbf{E}_{x_1, \ldots, x_n} [p] = \littlesum_{i=1}^n a_{ii}$
and hence
$$
\left(\mathbf{E}_{x_1, \ldots, x_n} [p] \right)^2 =
\littlesum_{i=1}^n a_{ii}^2 + \sum_{1 \le i < j \le n}  2 a_{ii} a_{jj}
$$
Hence, $\Var(p) = \sum_{i=1}^n (b_i^2 + 2 a_{ii}^2) + \sum_{1 \le i < j \le n} a_{ij}^2$. This implies the claimed bounds.
\end{proof}

\begin{claim}\label{clm:construction}
For the polynomial $q(y_1,x)$ constructed in Definition~\ref{def:residue},
the distributions of $q(y_1,x)$ and $p(x)$ are identical when $(y_1, x_1,
\ldots, x_n) \sim \mathcal{N}(0,1)^n$.
\end{claim}

\begin{proof}
Note that
$$
p(x) = p(\alpha_1 L_1(x) + R_1(x), \ldots, \alpha_n L_1(x)  +R_n(x)).
$$
Let $D$ be the joint distribution of $(R_1(x), \ldots, R_n(x))$
when $(x_1, \ldots, x_n) \sim \mathcal{N}(0,1)^n$.
As $R_i(x)$ is orthogonal to $L_1(x)$ for all $i \in [n]$, hence $D$ is
independent of the distribution of $L_1(x)  = \littlesum_{i=1}^n w_i x_i$
(when $(x_1, \ldots, x_n) \sim \mathcal{N}(0,1)^n$). Also, $L_1(x)$
is distributed like a standard normal.
Using these two facts, we get the claimed statement.
\end{proof}

\begin{claim}\label{clm:residue}
Given a degree-$2$ polynomial $p : \mathbb{R}^n \rightarrow \mathbb{R}$,
let $L_1(x)$ be a normalized linear form and $A$ be the matrix corresponding to
the quadratic part of $p$. Let $w_1$ be the vector corresponding to $L_1(x)$
and let $Q$ be any orthonormal matrix whose first column is $w_1$.
Let $\tilde{A} = Q^T \cdot A \cdot Q$. Then $\Var(\mathop{Res}(p,L_1(x))) =
4 \cdot \littlesum_{1 \le j \leq n}
\tilde{A}_{1j}^2$.
\end{claim}

\begin{proof}
Let $x_i = \alpha_{i1} L_1(x) + R_i(x)$.
Let $Q = [w_1, \ldots, w_n]$ be an orthonormal matrix. Let $L_i(x)$ be the linear form corresponding to $w_i$. Note that $x_i = \littlesum_{j=1}^n Q_{ij} L_j(x)$ and hence $R_i(x) = \littlesum_{j>1}^n Q_{ij} L_j(x)$. Since $L_1(x), \ldots, L_n(x)$ are orthonormal, hence their joint distribution is same as $(y_1, \ldots, y_n) \sim \mathcal{N}(0,1)^n$.

Also, observe that this implies that the joint distribution of $R_1(x), \ldots, R_n(x)$ is independent of $L_1(x)$. As a consequence, we get that the distribution of $\mathop{Res}(p, L_1(x))$ is same as $\mathop{Res}(\tilde{p}, y_1)$ where
$$
\tilde{p}(y_1, \ldots, y_n) = p(\littlesum_{j=1}^n Q_{1j} y_j , \ldots, \littlesum_{j=1}^n Q_{nj} y_j) =p((Q\cdot y)_1 , \ldots, (Q \cdot y)_n)
$$
Note that since $A$ is the matrix corresponding to the quadratic part of $p$, we get that
$$
\mathop{Res}(\tilde{p}, y_1) = \mathop{Res}(p((Q\cdot y)_1 , \ldots, (Q \cdot y)_n),y_1) = \mathop{Res}(y^T \cdot Q^T \cdot A \cdot Q \cdot y, y_1) = \mathop{Res} (y^T \cdot \tilde{A} \cdot y, y_1)
$$
Thus, $\mathop{Res}(y^T \cdot \tilde{A} \cdot y, y_1) = \littlesum_{j=2}^n 2\tilde{A}_{1j} y_1 y_j$. Thus, $\Var (\mathop{Res}(y^T \cdot \tilde{A} \cdot y, y_1) = 4 \cdot \littlesum_{1 \le j \leq n} \tilde{A}_{1j}^2$.
\end{proof}


\subsubsection{Proof of Theorem~\ref{thm:degree-decomp}.}
\label{sec:pf-thm-degree-decomp}

Recall the statement of Theorem~\ref{thm:degree-decomp}:

\medskip
\noindent
{\bf Theorem~\ref{thm:degree-decomp}.}
\emph{
Let $p : \mathbb{R}^n \rightarrow \mathbb{R}$ be a degree-$2$ polynomial
(with constant term $0$) whose entries are $b$-bit integers and
let $\epsilon, \eta>0$.  There exists a deterministic algorithm
\textsf{APPROXIMATE-DECOMPOSE} which on input an explicit description
of $p$, $\eps$ and $\eta$ runs in time $\poly(n,b,1/\eps,1/\eta)$ and
has the following guarantee :
\begin{itemize}
\item[(a)] If $\lambda_{\max}(p) \ge \eps \sqrt{\Var(p)}$, then
the algorithm outputs rational numbers $\lambda_1$, $\mu_1$ and a
degree-$2$ polynomial $r : \mathbb{R}^{n+1} \rightarrow \mathbb{R}$ with the
following property: for
$(y,x_1,\dots,x_n) \sim \mathcal{N}(0,1)^{n+1}$, the
   two distributions $p(x_1,\dots,x_n)$ and $q(y_1,x_1,\dots,x_n)$ are
   identical, where $q(y_1,x_1,\dots,x_n)$ equals $\lambda_1 y_1^2 + \mu_1
   y_1 + r(y_1,x_1,\dots,x_n).$
Further, $\Var(\mathop{Res}(r,y_1)) \le 4\eta^2 \Var(p)$ and $\Var(r)
\le (1-\eps^4/40) \cdot \Var(p)$.
\item[(b)] If $\lambda_{\max}(p) < \eps \sqrt{\Var(p)}$, then the algorithm
either outputs ``small max eigenvalue" or has the same guarantee as (a).
\end{itemize}
}

\medskip

\begin{proof}
The algorithm works as follows :
\begin{itemize}

\item[(i)] Let $A$ be the matrix corresponding to the quadratic part of $p$.
If $\Vert A \Vert_F^2 < \eps^2 \cdot \Var (p)$, then output ``small max
eigenvalue".

\item[(ii)] Run the algorithm \textsf{APPROXIMATE-LARGEST-EIGEN} from
Theorem~\ref{thm:primitive} on the matrix $A$. If the output is ``small max
eigenvalue", then return ``small max eigenvalue".

\item[(iii)] If the output of the algorithm \textsf{APPROXIMATE-LARGEST-EIGEN}
is the tuple $(\lambda, w_1)$, then for each unit vector $e_i$, we
express $e_i = \alpha_i w_1 + v_i$ where $v_i$ is orthogonal to $w_1$.  For the
sake of brevity, we will henceforth refer to $\littlesum_{j=1}^n v_{ij} x_j$
as $R_i(x)$.

\item[(iv)] Define the polynomial $q(y_1, x_1, \ldots, x_n)$ as $p(\alpha_1 y_1 + R_1(x), \ldots, \alpha_n y_1 + R_n(x))$.



\end{itemize}

The bound on the running time of the algorithm is obvious.
We now give the proof of correctness of the algorithm. First of all,
if $\Vert A \Vert_F^2 < \eps^2 \cdot \Var (p)$, then $\lambda_{\max}^2(A)
\le \Vert A \Vert_F^2 < \eps^2 \cdot \Var (p)$ and hence the output is correct.
So from now on, we assume that $\Vert A \Vert_F^2 \geq \eps^2 \cdot \Var(p)$.

Now, assuming that the output of Theorem~\ref{thm:primitive} in Step (ii) is
``small max eigenvalue", then by the guarantee of Theorem~\ref{thm:primitive}
it must be the case that $\lambda_{\max}^2 (A) \le \epsilon^2
\cdot \Vert A \Vert_F^2$. Using Fact~\ref{fac:quad-norm}, we get
that $\lambda_{\max}^2 (A) \le \epsilon^2 \cdot \Var(p)/2$.

Thus, in both the cases that the output of the algorithm is ``small max
eigenvalue", it is the case that $\lambda_{\max}^2 (A) \le \epsilon^2
\cdot \Var(p)$.

Now, consider the case in which the algorithm reaches Step (iii).
It must be the case that
$\Vert A \Vert_F \ge \eps \cdot \sqrt{\Var(p)}$,
and by Claim~\ref{clm:eigen-1}
it must hold that $\lambda_{\max}(A)
\geq (\eps/2) \cdot \Vert A \Vert_F$.  We start with the following claim.

\begin{claim}
The distribution of $q(y_1,x_1,\ldots, x_n)$ and $p(x_1,\ldots, x_n)$
are identical when $(y_1,x_1,\ldots, x_n) \sim \mathcal{N}(0,1)^{n+1}$.
\end{claim}

\begin{proof}
The polynomial $q(y_1, x_1, \ldots, x_n)$ is the same as the one
constructed in Definition~\ref{def:residue} and hence we can use
Claim~\ref{clm:construction} to get the stated claim.
\ignore{


END IGNORE}
\end{proof}
Next, we prove the bound on $\Var  (\mathop{Res}(q, y_1))$.
\begin{claim}\label{clm:bound-res}
$\Var  (\mathop{Res}(q, y_1)) \le 2 \eta^2  \Var(p)$.
\end{claim} \begin{proof}
Consider the orthogonal matrix $Q = [w_1, \ldots, w_n]$. Then, by using Claim~\ref{clm:residue}, for $\tilde{A} = Q^T \cdot A \cdot Q$,
$
\Var   (\mathop{Res}(q, y_1)) = 4 \littlesum_{j>1} \tilde{A}_{1j}^2
$. Now, using that $w_1, \ldots, w_n$ form an orthonormal basis, we get
$$
\sum_{j>1} \tilde{A}_{1j}^2 = \sum_{j>1} (w_j^T \cdot A \cdot w_1)^2
$$
Now, note that
for the value $\lambda$ that
\textsf{APPROXIMATE-LARGEST-EIGEN} outputs in Step~(iii) of
\textsf{APPROXIMATE-DECOMPOSE}, we have
$$
\eta^2 \Vert A \Vert_F^2 \ge \Vert A \cdot w_1  -
{\lambda} w_1 \Vert_2^2 = \sum_{j=1}^n (w_j^T \cdot A \cdot w_1 - \lambda w_j^T \cdot w_1)^2  \ge \sum_{j=2}^n (w_j^T \cdot A \cdot w_1)^2
$$
Here the first inequality follows from Theorem~\ref{thm:primitive}
part (a)(ii) while the second equality follows from the definition of $\ell_2$ norm of a vector. Thus, we get that
$$
\Var   (\mathop{Res}(q, y_1)) = 4 \littlesum_{j>1} \tilde{A}_{1j}^2 \le4 \cdot \Vert A \cdot w_1  - \lambda_1 w_1 \Vert_2^2  \le4 \eta^2 \Vert A \Vert_F^2 \le 2 \eta^2 \Var(p)
$$
The last inequality uses Fact~\ref{fac:quad-norm}.
\end{proof}

Thus, the only part that remains to be shown is that the variance of $r$ goes down.
\begin{claim}
$\Var(r (y_1,x_1,x_2, \ldots, x_n)) \le (1- \eps^4/40) \cdot \Var(p).$
\end{claim}
\begin{proof}
Note that $q(x) = \lambda_1y_1^2 + \mu_1 y_1 + r(y_1,x_1, \ldots, x_n)$. Let $\tilde{r}(y_1, x_1, \ldots, x_n) =\mathop{Res}(r,y_1)$. Since $\Var(\tilde{r}) \le 2 \eta^2 \Var(p)$, hence using Fact~\ref{fac:variance-difference}, we get that
$$
\Var( \lambda_1 y_1^2 + \mu_1 y_1 + r (y_1,x_1,x_2, \ldots, x_n) -
\tilde{r}(y_1, x_1, \ldots, x_n) ) \le (1+2\eta)^2 \Var(p).
$$
However, note that $r (y_1,x_1,x_2, \ldots, x_n) - \tilde{r}(y_1, x_1, \ldots, x_n)$ is independent of $y$ (call it $\tilde{r}_1(x_1, \ldots, x_n)$).
As a result, we get
$$
\Var(\tilde{r}_1(x_1, \ldots, x_n)) +\Var( \lambda_1 y_1^2 + \mu_1 y_1) \le (1+2\eta)^2 \Var(p)
$$
Since the algorithm reaches Step (iii) only if $\lambda_{\max}(A) \ge
\eps \Vert A \Vert_F/2$ and $\Vert A \Vert_F \ge
\eps \sqrt{\Var(p)}$, hence $\lambda_1 \ge (1-\eta) (\eps^2/2)
\sqrt{\Var(p)}$ and hence
$$
\Var(\tilde{r}_1(x_1, \ldots, x_n))  \le  (1+2\eta)^2 \Var(p)  - (1-\eta)^2
(\eps^4{/4}) \Var(p) \le (1- \eps^4/20) \cdot \Var(p).
$$
This uses the fact $\eta \le \eps^4/10^8$. Again, using Fact~\ref{fac:variance-difference} and that $\eta \le \eps^4/10^8$, we get that since $\Var(\tilde{r}) \le 2 \eta^2 \Var(p)$
$$
\Var(r (y_1,x_1,x_2, \ldots, x_n)) \le (1- \eps^4/40) \cdot \Var(p).
$$
\end{proof}

This concludes the proof of Theorem~\ref{thm:degree-decomp}.
\end{proof}
\fi

\subsection{The first stage:  Constructing a junta polynomial.}\label{sec:first-stage}~
In this section, we describe an algorithm {\tt Construct-Junta-PTF}  
which given as input an $n$ variable quadratic polynomial $p$, runs in time $\poly(n/\eps)$ and outputs a quadratic polynomial $q$ on $\tilde{O}(1/\eps^4)$ variables such that the distributions of $p(Y)$ and $q(Y)$ (when $Y\sim \mathcal{N}(0,1)^n$) are $O(\eps)$ close in Kolmogorov distance.  More precisely, we prove the following theorem.

\begin{theorem} \label{thm:construct-junta-PTF}
The algorithm {\tt Construct-Junta-PTF} has the following performance
guarantee: It takes as input an explicit description of an $n$-variable
degree-$2$ polynomial $p$ with {$b$-bit} integer coefficients, and
a value $\eps > 0.$ It runs (deterministically) in time
$\poly(n,b,1/\eps)$ and outputs a degree-$2$
polynomial $q = \sum_{i=1}^K (\lambda_i y_i^2 + \mu_i y_i) + C'$ such that
$
\left|\Pr_{x \in \mathcal{N}(0,1)^n}[p(x) \geq 0] - \Pr_{y \in \mathcal{N}(0,1)^K}[q(y) \geq 0]
\right| \leq O(\eps),
$
where
each $\lambda_i, \mu_i \in \Z$ and $K = \tilde{O}(1/\eps^{4})$.
\end{theorem}


\begin{figure}[tb]
\hskip-.2in \framebox{
\medskip \noindent \begin{minipage}{16.5cm}

\medskip

{\tt Construct-Junta-PTF}

\medskip

\noindent {\bf Input:}  Explicit description of an $n$-variable
degree-$2$  polynomial $p$ and $\eps >0.$\\
\noindent {\bf Output:} A degree-2 polynomial $q(y) = \sum_{i=1}^K (\lambda_i y_i^2 + \mu_i y_i)+C'$,
with $K = \tilde{O}(1/\eps^{{4}})$, such that
$
|\Pr_{x \in \mathcal{N}(0,1)^n}[p(x) \geq 0] - \Pr_{y \in \mathcal{N}(0,1)^K}[q(y) \geq 0]
| = O(\eps).
 $

\medskip

Let $p(x) = p'(x)+C$, where $p'$ has constant term $0$.
Assume by rescaling that $\Var(p')= \Var(p)= 1$.

\smallskip

{\bf Initialize:} $i=1$; $s_1(x) = p'(x)$; $h_0 \equiv 0$.

\smallskip

\noindent {\bf Repeat the following:}

\begin{enumerate}
\item Fix $\alpha \eqdef \Theta(\eps^4 / \log^2(1/\eps))$.
\item {\bf If} $\Var(s_i) < \alpha$ {\bf output} the polynomial $q_{i-1}:\R^{i-1} \to \R$ defined by
$q_{i-1}(y_1, \ldots, y_{i-1}) = h_{i-1}(y_1, \ldots, y_{i-1}) + \E[s_{i}] + C.$
\item {\bf If} $\Var(s_i) \ge \alpha$ {\bf do the following}
\begin{enumerate}
\item[(a)] Round down each coefficient of $s_i$ to its closest integer multiple of $\gamma/(Kn)$,
where {$\gamma = \tilde{\Theta}((\eps/K)^2) \sqrt{\alpha}$}
and {$K = \tilde{\Theta}(1/\eps^4)$}.
Let $s'_i(x)$ be the rounded polynomial.
\item[(b)] Run the routine \textsf{APPROXIMATE-DECOMPOSE}$(s'_i, \eps, \eta:={\tilde{\Theta}(\eps^4/K^4)})$.
\item[(c)] {\bf If} the routine returns ``small max eigenvalue" {\bf output} the polynomial $q'_{i}:\R^{i} \to \R$
$q'_i(y_1, \ldots, y_i) = h_{i-1}(y_1, \ldots, y_{i-1}) +  \beta_{s'_i} y_i + \E[s'_i]+C.$
where $\beta_{s'_i}$ is obtained by rounding down $\sqrt{\Var(s'_i)}$ to the nearest integer multiple of
{$\epsilon \alpha / 2$}.
\item[(d)] {\bf If} the routine outputs numbers $\lambda_i, \mu_i$ and a polynomial $r_i: \R^{n+1} \to \R$,
we define
$
s_{i+1}(x)  = r_i(y_i,x_1, \ldots, x_n) - \mathop{Res}(r_i , y_i)
$
and
$
h_i(y_1, \ldots, y_i) = h_{i-1}(y_1,\ldots, y_{i-1}) + (\lambda_i y_i^2 + \mu_i y_i).
$
\end{enumerate}
\item $i=i+1.$



%
\end{enumerate}

\ifnum\confversion=0
\noindent {\bf End Loop}
\fi



\end{minipage}}
\label{fig:aa}
\end{figure}
\medskip


The full proof of the theorem is technical
\ifnum\confversion=1
and is deferred to the  full version. Here we give some intuition behind the algorithm and its
proof of correctness.
\else
so first we give some intuition behind the algorithm and its
proof of correctness.
\fi

As mentioned in the introduction, if we were given the exact SVD
then we could construct a decoupled $n$-variable polynomial $\tilde{p}$ such
that $p(X)$ and $\tilde{p}(X)$ have the same distribution when
$X \sim \mathcal{N}(0,1)^n$.
Since we cannot compute the exact SVD, we instead iteratively use the
\textsf{APPROXIMATE-DECOMPOSE} algorithm.

Consider the first time the \textsf{APPROXIMATE-DECOMPOSE} algorithm
is called (on the polynomial $s'_1$ -- think of this as just
being the input polynomial $p$).  If it outputs ``small max
eigenvalue'', then using Chatterjee's recent CLT for functions of Gaussians,
we show that for $X \sim \mathcal{N}(0,1)^n$
the distribution of $s'_1(X)$ is close to
$\mathcal{N}(\mathbf{E}[s'_1], \Var(s'_1))$ in total variation distance.
In this case we can set the polynomial $q = \sqrt{\Var[s'_1]}
y_1 +  \mathbf{E}[s'_1]$
(note that we ignore technical details like the ``rounding'' that the
algorithm performs in this intuitive discussion).

On the other hand,
if the \textsf{APPROXIMATE-DECOMPOSE} algorithm does not output ``small max
eigenvalue'', then let $\lambda_1 y_1^2 + \mu_1 y_1 + r_1(y_1, x_1,
\ldots, x_n)$ be the output of \textsf{APPROXIMATE-DECOMPOSE}.
Since $\Var(\mathop{Res}(r,y_1))$ is small, it is not difficult to show
that the distribution of $\lambda_1 y_1^2 + \mu_1 y_1 + r_1(y_1, x_1,
\ldots, x_n) - \mathop{Res}(r,y_1)$ is close to the distribution of $s'_1$
in Kolmogorov distance. However, note that by definition,
$r_1(y_1, x_1, \ldots, x_n) - \mathop{Res}(r_1,y_1)$ does not
involve the variable $y_1$.
We now iteratively work with the polynomial
$s_2(x_1, \ldots, x_n) = r_1(y_1, x_1, \ldots, x_n) - \mathop{Res}(r_1,y_1)$.
In particular, we
apply  \textsf{APPROXIMATE-DECOMPOSE} to the polynomial
$s'_2$ (this is a ``rounded'' version of $s_2$ -- think of it as
just being $s_2$).
In this second stage, if \textsf{APPROXIMATE-DECOMPOSE} outputs
``small max eigenvalue'', we can set the polynomial $q$ to be
$q= \lambda_1 y_1^2 + \mu_1 y_1 + \sqrt{\Var(s'_2)} y_2 +
\mathbf{E}[s'_2]$; otherwise we proceed as we did earlier
to obtain $\lambda_2$, $\mu_2$, $r_2(y_2,x_1,\dots,x_n)$; and so on.

If this iterative procedure terminates within
$\tilde{O}(1/\epsilon^4)$ steps because of
\textsf{APPROXIMATE-DECOMPOSE} returning ``small max eigenvalue'' at some
stage, then its output is easily seen to satisfy the conditions of the
theorem.
If the procedure continues through $K=\tilde{O}(1/\epsilon^4)$ steps
without terminating, then using the critical-index style analysis,
it can be shown that the variance of the remaining polynomial $s_K$
is at most $O(\eps)$
(recall from Theorem \ref{thm:degree-decomp} that each call to
\textsf{APPROXIMATE-DECOMPOSE} reduces the variance of the polynomial
by a multiplicative factor of $(1 - \eps^4/40)$).
Since this variance is so small it can be shown that the remaining
polynomial can be safely ignored and that the polynomial
$\sum_{1 \le i \le K} \lambda_i y_i^2 + \mu_i y_i + \mathbf{E}[s_K]$
satisfies the conditions of the theorem.


\medskip

\ifnum \confversion=0
\noindent The rest of this section is devoted to the proof of the above
theorem.  We start with a couple of preliminary lemmas:

\begin{lemma}\label{lem:low-variance-kolmogorov}
Let $ p,q :\mathbb{R}^n \rightarrow \mathbb{R}$ be degree-$2$ polynomials
such that $\Var(q) \le \eps' \Var(p)$, where
{$\eps' = O(\eps^4 / \log^2(1/\eps))$}.
For $x \sim \mathcal{N}(0,1)^n$, let $D$ denote the distribution of $p(x) + q(x)$
and $\tilde{D}$ the distribution of $p(x) + \E[q(x)]$. Then we have
$d_K(D, \tilde{D}) \le \eps$.
\end{lemma}
\begin{proof}
By the definition of the Kolmogorov distance it is no loss of
generality to assume that $\E[q]=0.$ For a fixed but arbitrary
$\theta \in \R$ we will show
that
\[ \left| \Pr[p(x)+q(x) \le \theta] - \Pr[p(x) \le \theta]  \right| \le \eps.\]
We bound the LHS as follows: Fix $x \in \R^n$. The point $x$
contributes to the LHS only if
there exists $\delta>0$ such that either $|p(x) - \theta| \le \delta$ or $|q(x)| \ge \delta.$ We select $\delta$ appropriately and
bound the probability of the first event using Theorem~\ref{thm:cw} and the probability of the second event using Theorem~\ref{thm:dcb}.
 Indeed, fix $\delta  = \Theta(\eps^2) \sqrt{ \Var(p)} \ge \Omega (\log(1/\eps)) \sqrt{\Var(q)}$, where the inequality follows from the assumption
 $\Var(q) \le \eps' \Var(p)$.
Theorem~\ref{thm:cw} yields
\begin{equation} \label{eqn:cw-ub}
\Pr_{x \sim \mathcal{N}(0,1)^n} \left[ |p(x) - \theta| \le \delta \right]
= \Pr_{x \sim \mathcal{N}(0,1)^n} \left[ |p(x) - \theta| \le  \Theta(\eps^2) \sqrt{
\Var(p)}  \right]
\le \eps/2
\end{equation}
and by Theorem~\ref{thm:dcb}
\begin{equation} \label{eqn:cb-ub}
\Pr_{x \sim \mathcal{N}(0,1)^n} \left[|q(x) | \ge \delta \right] \le
\Pr_{x \sim \mathcal{N}(0,1)^n} \left[|q(x) | \ge  \Omega (\log(1/\eps)) \sqrt{\Var(q)}
\right] \le \eps/2.
\end{equation}
The lemma now follows by a union bound.
\end{proof}

\noindent As a consequence we have the following:
\begin{proposition} \label{prop:round}
Let $p: \R^n \to \R$ be a degree-$2$ polynomial with $\Var(p) {\ge \alpha}$.
Consider the polynomial $p': \R^n \to \R$ obtained by rounding down
each coefficient of $p$ to its closest integer multiple of $\gamma/n$, where {$\gamma = O (\eps^2/\log(1/\eps)) \cdot \sqrt{\alpha}$}. Then, we have
\[ \dk (p, p') \le \eps.\]
\end{proposition}

\begin{proof}
Note that $e(x) = p(x) - p'(x) =  \littlesum_{i\le j} \delta_{i,j} x_ix_j + \littlesum_{i} \gamma_i x_i$ where $|\delta_{i,j}| \le \gamma/n$ and $|\gamma_i| \le \gamma/n$.
As a consequence, we have $\mathop{SS}(e) \le \gamma^2$ and therefore
$$\Var(e) \le 2SS(e) \le \Theta (\eps^4 / \log^2(1/\eps)) \alpha \le  \Theta (\eps^4 / \log^2(1/\eps)) \Var(p)$$
where we used Claim~\ref{clm:variance-bound-SS} and the definition of $\gamma$.
{Fact~\ref{fac:variance-difference} gives that
$\Var(e) \leq \Theta (\eps^4 / \log^2(1/\eps)) \Var(p')$} and
Lemma~\ref{lem:low-variance-kolmogorov} now implies that
\[ \dk (p, p') \le \eps.\]
\end{proof}

For the sake of intuition, we start by analyzing the first
iteration of the loop.
In the beginning of the first iteration, we have $s_1(x) = p'(x)$,
where $p'$ is the polynomial $p$ without its constant term $C$.
{Hence} we have $\Var(s_1) = 1$, which means that Step~2 of the loop
is not executed. In Step~3(a) we round $s_1$ to obtain
the rounded polynomial $s'_1$. By Proposition~\ref{prop:round}, it follows that
\begin{equation} \label{eq:dksvssprime}
\dk(s_1, s'_1) \le \eps/K.
\end{equation}
Also note that the coefficients of $s'_1$ are integer multiples
{of $\gamma/(Kn)$} of magnitude $\poly(n/\eps)$,
{hence up to a scaling factor they are $\ell$-bit integers for
$\ell = O(\log(n/\eps)).$}
In Step~3(b) we run the routine
\textsf{APPROXIMATE-DECOMPOSE} on the rounded polynomial $s'_1$.
(Note that the routine runs in $\poly(n/\eps)$ time.)

If the routine returns ``small max eigenvalue'' (Step~3(c))  then Theorem~\ref{thm:degree-decomp} guarantees that the maximum magnitude eigenvalue
of $s'_1$ is indeed small, in particular $|\lambda_{\max}(s'_1)| \le \eps
\sqrt{\Var(s'_1)}.$ In this case, the algorithm outputs the univariate
polynomial $q'_1(y_1) = \beta_{s'_1} y_1 + \E[s'_1]+C$, where $ |\beta_{s'_1} - \sqrt{\Var(s'_1)}| \le {\eps\alpha/2}$.
We have the following:

\begin{claim} \label{claim:gaussian-first-iteration}
If $|\lambda_{\max}(s'_1)| \le \eps\sqrt{\Var(s'_1)}$, we have that
$\dk (s'_1(x),  q'_1(y_1)) = O(\eps).$
\end{claim}

To prove this claim we will need the following lemma. Its proof uses a powerful version of the Central Limit Theorem for functions of independent Gaussian random variables (which can be obtained using Stein's method):

\begin{lemma} \label{lem:stein}
Let $p:\R^n \to \R$ be a degree-$2$ polynomial over independent standard Gaussians.
If $|\lambda_{\max}(p)| \le \eps \sqrt{\Var(p)}$, then $p$ is $O(\eps)$-close to the Gaussian
$\mathcal{N}(\E[p], \Var(p))$ in total variation distance (hence, also in Kolmogorov distance).
\end{lemma}
The proof of Lemma~\ref{lem:stein} is deferred to Section~\ref{ssec:stein}.
\begin{proof}[Proof of Claim~\ref{claim:gaussian-first-iteration}]
Since $|\lambda_{\max}(s'_1)| \le \eps\sqrt{\Var(s'_1)}$, by Lemma~\ref{lem:stein} it follows that
$$\dk\left(s'_1(x), \sqrt{\Var(s'_1)}y_1+\E[s'_1]\right) = O(\eps),$$ where $x \sim \mathcal{N}(0,1)^n$ and $y_1 \sim \mathcal{N}(0,1)$.
Since $\beta_{s'_1} \le \sqrt{\Var(s'_1)}$, Fact~\ref{fact:gr} yields
$$\dk \left( \beta_{s'_1} y_1+\E[s'_1], \sqrt{\Var(s'_1)}y_1+\E[s'_1] \right) \le (1/2) \frac{| \beta^2_{s'_1}  - \Var(s'_1)|}{\beta^2_{s'_1}}
\le  \frac{\eps \alpha}{\beta^2_{s'_1}} = O(\eps)$$
where we used that $| \beta^2_{s'_1} - \Var(s'_1)| \le 2 | \beta_{s'_1} - \sqrt{\Var(s'_1)}|   \le \eps\alpha$
and that $\beta^2_{s'_1} \ge \alpha/2$ (which uses that $\Var(s_1) \ge \alpha$ and $|\Var(s_1) - \Var(s'_1)| \le 2\gamma^2/K^2 \le \eps \alpha$).

The claim follows from the aforementioned and the triangle inequality.
\end{proof}

{Now we analyze the execution of Step~3(d).}\ignore{
Theorem~\ref{thm:degree-decomp} guarantees that the maximum magnitude
eigenvalue of $s'_1$ is not too small, in particular $|\lambda_{\max}(s'_1)|
\ge {\eps/2} \sqrt{\Var(s'_1)}.$}
Consider the numbers $\lambda_1, \mu_1$
and degree-$2$ polynomial $r_1:\R^{n+1} \to \R$
returned by the routine Approximate-Decompose. Consider the polynomial
$$g_1(y_1, x_1, \ldots, x_n) = \lambda_1 y_1^2 +
\mu_1y_1+ r_1(y_1, x_1, \ldots, x_n).$$
Theorem~\ref{thm:degree-decomp} guarantees that the random variables $s'_1(x_1, \ldots, x_n)$ and $g_1(y_1, x_1, \ldots, x_n)$,
with $(y_1, x_1, \ldots, x_n) \sim \mathcal{N}(0,1)^{n+1}$
have identical distributions. (In particular, this implies that $\Var(g_1) = \Var(s'_1).$)
The algorithm proceeds to define
\[ s_2(x_1, \ldots, x_n) = r_1(y_1, x_1, \ldots, x_n) - Res(r_1, y_1)\] and
\[ h_1 (y_1) = \lambda_1y^2_1 + \mu_1y_1.\]
Note that the two summands ($\lambda_1y_1^2+\mu_1y_1$ and $r_1(y_1, x_1, \ldots, x_n)$)
defining $g_1$ are correlated (because of the variable $y_1$).
An important fact is that if we subtract $Res(r_1, y_1)$ from the polynomial $r_1$, the distribution of the resulting polynomial remains close in Kolmogorov distance:
\begin{claim} \label{claim:kolmogorov-first-iteration}
We have that $\dk (s'_1, s_2+h_1) \le {\eps/K}.$
\end{claim}
\begin{proof}
Note that $$\dk (s'_1, h_1+s_2) =  \dk(g_1, h_1+s_2) = \dk (g_1, g_1 - Res(r_1, y_1)).$$
By Theorem~\ref{thm:degree-decomp} it follows that $$\Var(Res(r_1, y_1)) \le \eta^2 \Var(s'_1) =  \eta^2 \Var(g_1).$$
{Since $\E[Res(r_1,y_1)]=0$,}
by Lemma~\ref{lem:low-variance-kolmogorov} we get
that $\dk (g_1, s_2+h_1) \le {\eps/K}$
as desired.
\end{proof}

The advantage of doing this is that $s_2$ and $h_1$ are
independent random variables, {since $s_2$ does not depend on $y_1.$}
As a consequence, we also obtain the following:
\begin{fact} \label{fact:big-first-iteration}
${\Var(s_2+h_1) =} \Var(s_2)+\Var(h_1) \ge 1-\eps/K.$
\end{fact}
\begin{proof}
Note that $$g_1 = h_1 + r_1 = h_1+ s_2+ Res(r_1, y_1).$$
We have that $\Var(g_1) = \Var(s'_1) \ge (1-\eps^2/K^2) \Var(s_1)$
and $\Var(Res(r_1, y_1)) \le 4\eta^2\Var(s'_1).$
By Fact~\ref{fac:variance-difference} it follows that
\[ \Var(s_2+h_1) \ge (1-2\eta) \Var(g_1) = (1-2\eta) \Var(s'_1)
\ge (1-2\eta)(1-\eps^2/K^2) \Var(s_1) \]
which completes the proof {since $\Var(s_1)=1.$}
\end{proof}

We also have that the variance of the polynomial $s_2$
{is smaller than $\Var(s_1)$} by a multiplicative factor:
\begin{claim} \label{claim:variance-first-iteration}
We have $\Var(s_2) \le (1-\eps^4/40).$
\end{claim}
\begin{proof}
By Theorem~\ref{thm:degree-decomp} we know that
$$\Var(r_1) \le (1-\eps^4/40) \Var(s'_1) \le  (1-\eps^4/40)$$
where the second inequality used the fact that $\Var(s'_1) \le \Var(s_1) \le 1.$
Now note that $s_2$ is obtained from $r_1$ by removing a subset of its terms.
Therefore, $\Var(s_2) \le \Var(r_1)$ and the claim follows.
\end{proof}

{This concludes our analysis of the first iteration of the loop.}

\medskip

We are now ready to analyze a generic iteration of the loop.
{Many aspects of this analysis will be similar to our earlier analysis
of the first iteration.}
Consider the ${j}$-th iteration of the loop, where $j \ge 2$.
We can assume by induction that for all $i<j$ the following hold:
\begin{itemize}

\item[(a)] $\dk (s_i, s'_i) \le \eps/K$
{(for $j=2$ this holds by (\ref{eq:dksvssprime}));}

\item[(b)] $\dk (h_{i-1}+s_{i}, h_i+s_{i+1}) \le 2\eps/K$,
{(for $j=2$ this holds by (a) and
Claim~\ref{claim:kolmogorov-first-iteration});}

\item[(c)] $\Var(s_{i+1}) + \Var(h_i) \ge (1-\eps/K)^{i}$, and
{(for $j=2$ this holds by
Fact~\ref{fact:big-first-iteration}); and}

\item[(d)] $\Var(s_{i+1}) \le (1-\eps^4/40)^i$
{(for $j=2$ this holds by
Claim~\ref{claim:variance-first-iteration}).}

\end{itemize}

{We start by observing that $j-1 \le K$, i.e.,
the total number of iterations is at most $K+1$.
Indeed, by (d) above, for $i= K$ we will have $\Var(s_{i+1})
\le(1-\eps^4/40)^{K} < \alpha $
and the algorithm terminates in Step~2.}

In the beginning of the $j$-th iteration, we have the polynomial $s_j(x_1,
\ldots, x_n)$,
satisfying $\Var(s_j) \le (1-\eps^4/40)^{j-1}$ and the polynomial
$h_{j-1}(y_1, \ldots, y_{j-1})$.  If the variance has become very small,
we can ``truncate'' $s_{{j}}$ taking into account its expectation (Step~2). Consider the polynomial
\[
q_{j-1}(y_1, \ldots, y_{j-1}) = h_{j-1}(y_1, \ldots, y_{j-1}) + \E[s_j]+C.
\]
We have the following claim:

\begin{claim} \label{claim:small-variance-jth-iteration}
Suppose that $\Var(s_j) < \alpha$, where $\alpha \eqdef \Theta (\eps^4 / \log^2(1/\eps)).$ Then, we have that
\[ \dk(h_{j-1}+s_j+C, q_{j-1} ) \le \eps. \]
\end{claim}
\begin{proof}
The claim is equivalent to showing that
$\dk (h_{j-1}+s_j, h_{j-1}+\E[s_j]) \le \eps.$
Note that by Part (c) of the inductive hypothesis we have that $\Var(h_{j-1})
+\Var(s_j) \ge (1-\eps/K)^{j-1}$.
Since $\Var(s_j) {<} \alpha$ we get
that $$\Var(h_{j-1}) {>} (1-\eps/K)^{j-1} - \alpha \ge 1/2,$$
where the last inequality uses the fact that $j \le K+1$.
The claim follows by an application of
Lemma~\ref{lem:low-variance-kolmogorov}
for the polynomials $h_{j-1}$ and $s_j$.
\end{proof}

Combining the above claim with Parts (a) and (b) of the inductive hypothesis
and using the triangle inequality completes the correctness analysis
{of the algorithm in the case that it exits in Step~2.}

\medskip

\noindent We now consider the complementary case that $\Var(s_j) \ge \alpha$ (Step~3).
In Step~3(a) we round $s_j$ to obtain
the rounded polynomial $s'_j$. By Proposition~\ref{prop:round}, it follows that
\[ \dk(s_j, s'_j) \le \eps/K\]
establishing Part (a) of the inductive hypothesis for $i=j$.

Also note that the coefficients of $s'_j$ are integer multiples
{of $\gamma/(Kn)$}
of magnitude $\poly(n/\eps)$.
In Step~3(b) we run the routine
{\textsf{APPROXIMATE-DECOMPOSE}} on the rounded polynomial $s'_j$.
{(Note that the routine runs in $\poly(n/\eps)$ time.)}

If the routine returns ``small max eigenvalue'' (Step~3(c))  then Theorem~\ref{thm:degree-decomp} guarantees that the maximum magnitude eigenvalue
of $s'_j$ is indeed small, in particular $|\lambda_{\max}(s'_j)| \le \eps \sqrt{\Var(s'_j)}.$ In this case, the algorithm outputs the
polynomial $q'_j(y_1, \ldots, j_j) = h_{j-1}(y_1, \ldots, y_{j-1}) +  \beta_{s'_j} y_j + \E[s'_j]+C$, where $ |\beta_{s'_j} - \sqrt{\Var(s'_j)}| \le {\eps\alpha/2}$.
We have the following, {which is very similar to
Claim~\ref{claim:gaussian-first-iteration}:}

\begin{claim} \label{claim:gaussian-jth-iteration}
If $|\lambda_{\max}(s'_j)| \le \eps\sqrt{\Var(s'_j)}$, we have that
$\dk (s'_j(x) , \beta_{s'_j} y_j + \E[s'_j] ) = O(\eps).$
\end{claim}

\begin{proof}
Since $|\lambda_{\max}(s'_j)| \le \eps\sqrt{\Var(s'_j)}$, by Lemma~\ref{lem:stein} it follows that
$$\dk\left(s'_j(x), \sqrt{\Var(s'_j)}y_j+\E[s'_j]\right) = O(\eps),$$ where $x \sim \mathcal{N}(0,1)^n$ and $y_j \sim \mathcal{N}(0,1)$.
Hence, by Fact~\ref{fact:gr}, we get that
$$\dk \left( \beta_{s'_j} y_j+\E[s'_j], \sqrt{\Var(s'_j)}y_j+\E[s'_j] \right) \le (1/2) \frac{| \beta^2_{s'_j}  - \Var(s'_j)|}{\beta^2_{s'_j}}
\le  \frac{\eps \alpha}{\beta^2_{s'_j}} = O(\eps)$$
where the second inequality used the fact that $| \beta^2_{s'_j}  - \Var(s'_j) | \le 2\eps \alpha$ and the last uses the fact
that $\beta^2_{s'_j} \ge \alpha/2.$
The claim follows from the aforementioned and the triangle inequality.
\end{proof}

\noindent Our final claim for this case {(the case
that the algorithm exits in Step~3(c))} is the following:

\begin{claim} \label{claim:jth-iteration-small}
We have that $\dk(q'_j , p) = O(\eps).$
\end{claim}

\begin{proof}
First, recall that $\dk(s_j, s'_j) \le \eps/K$ and therefore by the above claim and triangle inequality we get
$\dk (s_j(x) , \beta_{s'_j} y_j + \E[s'_{{j}}] ) = O(\eps).$ A convolution argument (exploiting independence) now gives
that $\dk (h_{j-1}+ s_j , h_{j-1}+ \beta_{s'_j} y_j + \E[s'_j] ) = O(\eps).$
Combining the above with Parts (a) and (b) of the inductive hypothesis yields the claim by an application
of the triangle inequality.
\end{proof}

{Now we analyze the execution of Step~3(d).  To finish the proof it
suffices to show that the inductive hypotheses (a)--(d) all hold for $i=j.$}
\ignore{
Theorem~\ref{thm:degree-decomp} guarantees that the maximum magnitude
eigenvalue of $s'_j$ is not too small, in particular $|\lambda_{\max}(s'_j)|
\ge {\eps/2} \sqrt{\Var(s'_j)}.$}Consider the
numbers $\lambda_j, \mu_j$ and degree-$2$ polynomial $r_j:\R^{n+1} \to \R$
returned by the routine \textsf{APPROXIMATE-DECOMPOSE}.
Consider the polynomial
$$g_j(y_j, x_1, \ldots, x_n) = \lambda_j y_j^2 + \mu_jy_j+ r_j(y_j, x_1, \ldots, x_n).$$
Theorem~\ref{thm:degree-decomp} guarantees that the
random variables $s'_j(x_1, \ldots, x_n)$ and $g_j(y_j, x_1, \ldots, x_n)$,
with $(y_j, x_1, \ldots, x_n) \sim \mathcal{N}(0,1)^{n+1}$ have identical
distributions. (In particular, this implies that $\Var(g_j) = \Var(s'_j).$)
The algorithm proceeds to define
\[ s_{j+1}(x_1, \ldots, x_n) = r_j(y_j, x_1, \ldots, x_n) - Res(r_j, y_j)\] and
\[ h_j (y_1, \ldots, y_j) = \littlesum_{i=1}^j (\lambda_i y_i^2+\mu_i y_i) .\]
Note that the {two summands $\lambda_j y_j^2 + \mu_j y_j$ and
$r_j(y_j,x_1,\dots,x_n)$}
in $g_j$ are correlated (because of the variable $y_j$).
Similarly to the first iteration,
if we remove $Res(r_j, y_j)$, there is a very small change in the Kolmogorov distance :

\begin{claim} \label{claim:kolmogorov-jth-iteration}
We have that $\dk (s'_j,  \lambda_j y_j^2 + \mu_jy_j+s_{j+1}) \le {\eps/K}.$
\end{claim}
\begin{proof}
Note that $$\dk (s'_j,  \lambda_j y_j^2 + \mu_jy_j+s_{j+1}) =  \dk (g_j,  \lambda_j y_j^2 + \mu_jy_j+s_{j+1}) = \dk (g_j, g_j - Res(r_j, y_j)).$$
By Theorem~\ref{thm:degree-decomp} it follows that $$\Var(Res(r_j, y_j)) \le \eta^2 \Var(s'_j) =  \eta^2 \Var(g_j).$$
{Since $\E[Res(r_j,y_j)=0],$}
by Lemma~\ref{lem:low-variance-kolmogorov} we get that
$\dk (g_j, s_{j+1}+h_j) \le {\eps/K}$ as desired.
\end{proof}



As a corollary, we establish Part (b) of the induction for $i=j$.
\begin{corollary} \label{cor:jth-iteration-kol2}
We have that $\dk (h_{j-1}+s_j, h_j+s_{j+1}) \le 2\eps/K$.
\end{corollary}
\begin{proof}
By Claim~\ref{claim:kolmogorov-jth-iteration} and the fact that $\dk(s_j, s'_j) \le \eps/K$, we get that
$\dk (s_j,  \lambda_j y_j^2 + \mu_jy_j+s_{j+1}) \le 2\eps/K.$
By a convolution argument (exploiting independence), it follows that
$\dk (h_{j-1}+s_j,  h_{j-1}+ \lambda_j y_j^2 + \mu_jy_j+s_{j+1}) \le 2\eps/K$
which completes the proof.
\end{proof}

As a consequence, we also obtain the following, establishing Part (c) of the induction:
\begin{fact} \label{fact:big-jth-iteration}
$\Var(s_{j+1})+\Var(h_j) \ge {(1-\eps/K)^j}.$
\end{fact}

\begin{proof}
By definition we can write
$$ g_j+h_{j-1}=   h_j + r_j =  h_j + s_{j+1}+ Res(r_j, y_j).$$
We first claim that
$$\Var(h_j+s_{j+1}) \ge (1-2\eta) \Var(h_{j{-1}}+g_j).$$
Indeed, by Theorem~{\ref{thm:degree-decomp}},
we have that $\Var(Res(r_j, y_j)) \le
4\eta^2\Var(s'_j)$ and $$\Var(g_j+h_{j-1}) =  \Var(g_j)+\Var(h_{j-1})
\ge \Var(g_j) =\Var(s'_j)$$
where the first equality used independence.
The claim now follows by Fact~\ref{fac:variance-difference}.
Our second claim is that
\[
\Var(g_j+h_{j-1}) \ge (1-\eps^2/K^2) \Var(s_j+h_{j-1}).\]
Indeed, we can write
$$\Var(g_j+h_{j-1}) = \Var(s'_j)+\Var(h_{j-1}) \ge  (1-\eps^2/K^2) \Var(s_j)
+\Var(h_{j-1}) \ge (1-\eps^2/K^2) \Var(s_j+h_{j-1}).
$$
The desired fact follows by combining the above two claims with Part (c) of the inductive hypothesis.
\end{proof}

{Finally, we} show that the variance of the polynomial $s_{j+1}$ will
decrease by a multiplicative factor, {giving (d) and}
completing the induction.

\begin{claim} \label{claim:variance-jth-iteration}
We have $\Var(s_{j+1}) \le (1-\eps^4/40)^j.$
\end{claim}
\begin{proof}
By Theorem~\ref{thm:degree-decomp} we know that
$$\Var(r_j) \le (1-\eps^4/40) \Var(s'_j) \le  (1-\eps^4/40)\Var(s_j)
\le (1-\eps^4/40)^j$$
where we used the fact that $\Var(s'_2) \le \Var(s_2)$
and Part (d) of the induction hypothesis.
Now note that $s_{j+1}$ is obtained from $r_j$ by removing a subset of its terms.
Therefore, $\Var(s_{j+1}) \le \Var(r_j)$ and the claim follows.
\end{proof}

This completes the proof of correctness.
We claim that the algorithm runs in $\poly({n,b,1/\eps})$ time.
This follows from the fact that the number of iterations
of the loop is at most $K+1 = \poly(1/\eps)$ and each iteration runs in
$\poly({n,b,1/\eps})$ time. Indeed, it is easy to verify that the
running time of a given iteration is dominated by the call
to the \textsf{APPROXIMATE-DECOMPOSE} routine.
Since the input to this routine is the polynomial $s'_i$ whose coefficients
(up to rescaling) are integers whose magnitude is $\poly(n/\eps)$,
it follows that the routine runs in polynomial time.

\subsubsection{Proof of Lemma~\ref{lem:stein}.} \label{ssec:stein}
We note that even the Kolmogorov distance version of the lemma (which is sufficient for our purposes)
does not follow immediately from the Berry-Ess{\'e}en theorem. 
One can potentially deduce our desired statement from Berry-Ess{\'e}en  by using an appropriate
case analysis on the structure of the coefficients. However, we show that it can be deduced in a more principled way from a CLT version obtained
using Stein's method. In particular, we will need the following theorem of Chatterjee:

\begin{theorem} \label{thm:chat}
\cite{Chatterjee:09}
Let $X \sim \mathcal{N}(0,1)^n$ and $f : \mathbb{R}^n \rightarrow \mathbb{R}$. Let $W = f(X_1, \ldots, X_n)$. Suppose that $\E[W]=\mu$ and $\Var[W]=\sigma^2$.
{Let $Y \sim \mathcal{N}(0,1)^n$ be independent of $X$.}
Define the random variable $T(X)$ as
$$
T(X) = \int_{t=0}^1 \frac{1}{2\sqrt{t}} \cdot \E_{Y}  \left[ \littlesum_{i=1}^n \frac{\partial f(X)}{\partial X_i}  \cdot \frac{\partial f(\sqrt{t} X + \sqrt{1-t} Y)}{\partial X_i} \right] \ dt.$$
Then we have that $$\dtv \left( f(X), \mathcal{N}(\mu,\sigma^2) \right) \le \frac{\sqrt{\Var[ T]}}{\sigma^2}.$$
\end{theorem}

\noindent {Note that by Claim~\ref{clm:poly-equivalence}
we can assume that $p$ is of the form $p(x) = \littlesum_i (\lambda_i x_i^2 + \mu_i x_i)$.}
We want to apply this theorem to deduce that the random variable $p(X)$ with $X \sim \mathcal{N}(0,1)^n$ is $O(\eps)$ close to a Gaussian with the right mean
and variance. Note that $\E[p(X)] = \littlesum_{i=1}^n \lambda_i$ and $\Var[p(X)] = \littlesum_{i=1}^n (2\lambda_i^2+\mu_i^2)$. We will apply the above theorem
for the function $f(x) = p(x) = \littlesum_{i=1}^n (\lambda_i x_i^2 + \mu_i x_i).$
We have that $ \frac{\partial p(x)}{\partial x_i} = 2\lambda_ix_i + \mu_i.$
For $t \in [0,1]$ we can write $$p(\sqrt{t}x+\sqrt{1-t}y) =  \littlesum_{i=1}^n \lambda_i (\sqrt{t} x_i + \sqrt{1-t}y_i)^2 + \littlesum_{i=1}^n \mu_i (\sqrt{t} x_i + \sqrt{1-t}y_i)$$
and therefore
\[  \frac{ \partial p(\sqrt{t}x+\sqrt{1-t}y)}{\partial x_i} = \lambda_i (2tx_i + 2\sqrt{t(1-t)}y_i) + \mu_i \sqrt{t}.\]
Therefore,
\[
\E_{Y}  \left[ \littlesum_{i=1}^n \frac{\partial p(X)}{\partial X_i}  \cdot \frac{\partial p(\sqrt{t} X + \sqrt{1-t} Y)}{\partial X_i} \right] =
\littlesum_{i=1}^n (2\lambda_iX_i+\mu_i)(2\lambda_itX_i+\mu_i \sqrt{t})
\] and the desired integral equals
\begin{eqnarray*}
T = \littlesum_{i=1}^n (2\lambda_iX_i+\mu_i)(\lambda_i X_i \int_{0}^1 \sqrt{t} dt+\mu_i /2) &=&
 \littlesum_{i=1}^n (2\lambda_iX_i+\mu_i)(2\lambda_i X_i/3 +\mu_i /2)\\
  &=&  \littlesum_{i=1}^n (4\lambda_i^2 X_i^2/3 + 5/3\lambda_i \mu_iX_i + \mu_i^2/2)
\end{eqnarray*} from which it follows that
\begin{eqnarray*}
\Var[T] = \littlesum_{i=1}^n \left( 2 (4\lambda_i^2/3)^2 + (5/3\lambda_i \mu_i)^2 \right) &\le&  \littlesum_{i=1}^n \left( 2 (5\lambda_i^2/3)^2 + (5/3\lambda_i \mu_i)^2 \right) \\
 &=& (25/9) \littlesum_{i=1}^n \left( 2 \lambda_i^4 +  \lambda_i^2 \mu_i^2 \right)  \\
 &\le& (25/9) \max_i \lambda_i^2 \cdot \littlesum_{i=1}^n \left( 2 \lambda_i^2 + \mu_i^2 \right) \\
&\le& (25/9) (\eps^2 \Var[p]) \cdot \Var[p] \\
&=& (25/9) (\eps \Var[p])^2.
\end{eqnarray*}
An application of the Theorem
now yields that $\dtv\left( p(X), \mathcal{N}(\mu, \sigma^2) \right) \le 5\eps/3$.
Since $\dk(X, Y) \le \dtv(X,Y)$
for any pair of random variables $X, Y$ the lemma follows. \qed

\fi

\subsection{The second stage:
deterministic approximate counting for degree-2 junta PTFs
over Gaussians.} \label{sec:second-stage}

\ifnum\confversion=0
In this section we use the $K=\tilde{O}(1/\eps^4)$-variable polynomial $q(y)$ provided by Theorem~\ref{thm:construct-junta-PTF} to do efficient deterministic approximate counting.

One possible approach is to break the polynomial $q(y)$ into two components $q_+$ (corresponding
to those variables $y_i$ that have $\lambda_i > 0$) and $q_-$ (containing those $y_i$
that have $\lambda_i < 0$).  Both $q_+(y)$ and $-q_-(y)$ follow non-centered generalized chi-squared distributions, so it is conceivable that by directly analyzing the pdfs of such
distributions, one could (approximately) specify the region of $\R^K$
over which $q_+(y) - q_-(y) \geq 0$,
and then perform approximate numerical integration over that region to directly
estimate $\Pr_{y \sim \mathcal{N}(0,1)^K}[q(y) \geq 0]$.  While expressions have been given
for the pdf of a generalized chi-squared distribution without the linear part, we need expressions for the pdf when there is an additional linear part.  Even in the case where there is no
linear part,
the expressions for the pdf are somewhat forbidding (see equations (6) and (7) of
\cite{BHO09}), so an approach along these lines is somewhat unappealing.

Instead of pursuing this technically involved direction, we
adopt a technically and conceptually straightforward approach
based on simple dynamic programming.  The algorithm
{\tt Count-Junta}
that we propose and analyze is given below.
{Intuitively, the rounding that is performed
in the first step transforms the polynomial $q$ to one with
``small integer weights.'' This, together with the
discretization in Step~2 (which lets us approximate each independent
Gaussian input with a small discrete set of values), makes it possible
to perform dynamic programming to exactly count the
number of satisfying assignments of the corresponding PTF.}
\fi

\ifnum\confversion=1
In this section we use the $K=\tilde{O}(1/\eps^4)$-variable
polynomial $q$ provided by Theorem~\ref{thm:construct-junta-PTF}
to do efficient deterministic approximate counting.  The algorithm
{\tt Count-Junta}
that we propose and analyze is given below.
{Intuitively, the rounding that is performed
in the first step transforms the polynomial $q$ to one with
``small integer weights.'' This, together with the
discretization in Step~2 (which lets us approximate each independent
Gaussian input with a small discrete set of values), makes it possible
to perform dynamic programming to exactly count the
number of satisfying assignments of the corresponding PTF.}
\fi

\bigskip

\hskip-.2in \framebox{
\medskip \noindent \begin{minipage}{16.5cm}

\smallskip

{\tt Count-Junta}

\smallskip

\noindent {\bf Input:}  Explicit description of a
$K=\tilde{O}(1/\eps^{{4}})$-variable degree-$2$
polynomial $q(y) = \littlesum_{i=1}^K (\lambda_i y_i^2 + \mu_i y_i) + \tau$,
where each $\lambda_i,\mu_i, \tau \in \Z$;
parameter $\eps> 0.$

\noindent {\bf Output:} A value $v \in [0,1]$ such that
\ifnum\confversion=0
\[
\left|\Pr_{y \in \mathcal{N}(0,1)^K}[q(y) \geq 0] - v \right| \leq \eps.
 \]
\else
$
|\Pr_{y \in \mathcal{N}(0,1)^K}[q(y) \geq 0] - v | \leq \eps.
$
\fi

\medskip

\begin{enumerate}

\item {\bf Rounding.}  Set
$\eps' = \tilde{\Theta}(\eps^{{6}})$
to be a value of the form $1/2^{\tiny{\text{integer}}}$.
Let $M = \max\{|\lambda_1|,\dots,|\lambda_K|,
|\mu_1|,\dots,|\mu_K|\}.$  Let
$q'(y) = \littlesum_{i=1}^K (\lambda'_i y_i^2 + \mu'_i y_i) + {\tau'}$
be obtained from $q(y)$ by dividing all coefficients $\lambda_i,\mu_i,\tau$
by $2^{\lceil \log_2 M \rceil} \cdot (\eps'/2)$
and rounding the result to the nearest integer (so each of $\lambda_i{'},
\mu_i{'}$ is an integer with absolute value at most $2/\eps'$),

\ignore{Set $N=2/\eps'.$}

\ignore{
{If $|\tau'| > \tilde{O}(1/\eps^2)$ then output (1 if $\tau>0$, 0 if
$\tau < 0$); otherwise proceed to the next step.}
}

\ignore{
 rounding each coefficient of $q/\|q\|$ to
the nearest integer multiple of $\eps'/2$,
where $\eps' = \tilde{\Theta}(\eps)$
is of the form $1/2^{\tiny{\text{integer}}}$.
Write $\lambda'_i$ as $\ell_i/N$ and $\mu'_i$ as $m_i/N$ and $\tau'$ as $t/N$
where $N  = 2/\eps'$ is a positive integer.
}

\item {\bf Discretizing each coordinate.}
Set $\eps^\star = \Theta(\eps/K)$ to be of the
form $1/2^{\text{\tiny{integer}}}.$
For $i=1,\dots,K$:  run {\tt Discretize}$(\lambda'_i,\mu'_i,\ignore{N,}
\eps^\star)$ and
let $S_i = \{s_{i,1},\dots,s_{i,R}\}$
be the multiset that it returns.

\item {\bf Counting via dynamic programming.}
Run {\tt DP}$(S_1,\dots,S_K, \tau'\ignore{,N})$ and output the value it returns.

\end{enumerate}

\end{minipage}}

\bigskip

The performance guarantee of {\tt Count-Junta} is given in the following
theorem:

\begin{theorem} \label{thm:count-junta}
Algorithm {\tt Count-Junta} is given as input an explicit description of a
polynomial
$q(y) = \littlesum_{i=1}^K (\lambda_i y_i^2 + \mu_i y_i) + \tau$ and
$\eps>0$
where $\lambda_i,\mu_i \in \Z$, $K = \tilde{O}(1/\eps^{4})$,
{and each coefficient $\lambda_i, \mu_i, \tau$ is a $B$-bit
integer.}
It runs (deterministically) in $O({K B})\cdot \polylog(1/\eps) +
{\poly(1/\eps)}$ bit operations
and outputs a value $v \in
[0,1]$ such that
\ifnum\confversion=0
\begin{equation} \label{eq:close}
\left|\Pr_{y \sim \mathcal{N}(0,1)^K}[q(y) \geq 0] - v \right| \leq \eps.
\end{equation}
\else
$\left|\Pr_{y \sim \mathcal{N}(0,1)^K}[q(y) \geq 0] - v \right| \leq \eps.$
\fi
\end{theorem}

\ifnum\confversion=0
\subsubsection{Proof of Theorem~\ref{thm:count-junta}}

~

\medskip

\noindent {\bf Runtime analysis.}  It is straightforward to verify
that Step~1 (rounding) can be carried out, and the integers
$\lambda'_i, \mu'_i, \tau'$ obtained, in $O({K B}) \cdot \log(1/\eps)$ bit
operations (note that the coefficients $\lambda'_i, \mu'_i$ are
obtained from the original values simply by discarding all but the
$O(\log(1/\eps))$ most significant bits).
Note that each of $\lambda'_i, \mu'_i$ {is
a $O(\log(1/\eps))$-bit integer}.

The claimed running time of {\tt Count-Junta} then follows
easily from the running times established in Lemma~\ref{lem:discretize1}
and the analysis of {\tt DP} given below.

\medskip

\noindent
{\bf Correctness.}
\ignore{
{
We begin by observing that if $|\tau'| > \tilde{O}(1/\eps^2)$
(in which case the algorithm halts and outputs a value in $\{0,1\}$
at the end of Step~1) then it must be the case
that $|\tau| > O(\log(1/\eps))\Var(q).$
In this case correctness follows from the
``degree-$2$ Chernoff bound,'' Theorem~\ref{thm:dcb}.  So we henceforth assume
that $|\tau| \leq \tilde{O}(1/\eps^2).$
}}We start with a simple lemma establishing that the ``rounding'' step,
Step~1, does not change the acceptance probability of the PTF by more than a
small amount:

\begin{lemma} \label{lem:round-q}
We have
\begin{equation}\label{eq:round-q}
\left|
\Pr_{y \sim \mathcal{N}(0,1)^K}[q(y)\geq 0] - \Pr_{y \sim \mathcal{N}(0,1)^K}[q'(y) \geq 0] \right|
\leq \eps/2.
\end{equation}
\end{lemma}
\begin{proof}
This is a standard argument using concentration (tail bounds for degree-2
polynomials in Gaussian random variables) and anti-concentration
(Carbery-Wright).  Let $a(y) =
(2^{\lceil \log_2 M \rceil} \cdot (\eps'/2))q'(y) - q(y)$.
\ignore{\new{(note that the constant term of $a(y)$ is zero).}}
We have that $\sign(q(y)) \neq \sign(q'(y))$ only if at least one
of the following events occurs:
(i) $|a(y)| \geq c \eps^2 {\Var(q)}
\ignore{\|q\|_{{2}}}$, or
(ii) $|q(y)| \leq c \eps^2 {\Var(q)}$ (where $c$ is an
absolute constant).
For (i), we observe that $a(y)$ has at most $2K+1$
coefficients that are each at most $\eps' M$ in magnitude
{and hence $|\E_{y \sim \mathcal{N}(0,1)^K}[a(y)]| \leq (K+1)\eps'M$}
while ${\Var(q)} \geq M$
(recall that at last one of {$|\lambda_i|, |\mu_i|$} is at least $M$).
So ${\Var(a)} \leq \sqrt{2K+1} \cdot \eps' \cdot M$, and
by Theorem~\ref{thm:dcb} (the ``degree-2 Chernoff bound'')
$\Pr_{y \sim \mathcal{N}(0,1)^K}[|a(y)| \geq c \eps^2 {\Var(q)}]
\leq \eps/4$.
On the other hand, Theorem~\ref{thm:cw} gives us that $\Pr[
{|q(y)|} \leq c \eps^2 {\Var(q)}] \leq O(\sqrt{c \eps^2})
\leq \eps/4.$
The lemma follows {by a union bound}.
\end{proof}

We now turn to Step~2 of the algorithm, in which each distribution
$\lambda'_i y_i^2 + \mu'_i y_i${, $y_i \sim \mathcal{N}(0,1)$,} is converted to a nearby discrete distribution.
The procedure {\tt Discretize} is given below:

\medskip

\hskip-.2in \framebox{
\medskip \noindent \begin{minipage}{16.5cm}

\smallskip

{\tt Discretize}

\smallskip

\noindent {\bf Input:}
Integers $\ell,m\ignore{,N}$;
real value $\eps^\star = 1/2^{\tiny \text{integer}}$.

\noindent {\bf Output:} A multiset $S=\{s_1,\dots,s_R\}$
of $R=2/\eps^\star$ values such that
the distribution ${\cal D}_S$ satisfies
$
\dk({\cal D}_S, \ell Y^2 + m Y) \leq \eps^\star$,
where $Y \sim \mathcal{N}(0,1).
$

\medskip

\begin{enumerate}

\item Let $t_1 < \cdots < t_R$ be the real values given by
Fact~\ref{fact:cover} when its algorithm is run on input parameter
$\eps^\star.$

\item Output the multiset $S = \{
s_1,\dots,s_R\}$ where $s_i =
\ell t_i^2 + m t_i$ for
all $i$.

\end{enumerate}

\end{minipage}}

\bigskip

Our key lemma here says that ${\cal D}_S$ is Kolmogorov-close to the
univariate degree-2 Gaussian polynomial $\ell y_i^2 +
m y_i$:

\begin{lemma} \label{lem:discretize1}
Given $\ell,m,\eps^\star$ as specified in {\tt Discretize},
the procedure {\tt Discretize}$(\ell,m,\eps^\star)$
outputs a multiset $S=\{s_1,\dots,s_R\}$ of $R=8/\eps^\star$ values such that
the distribution ${\cal D}_S$ satisfies
\begin{equation} \label{eq:DS}
\dk({\cal D}_S, \ell Y^2 +  m Y) \leq \eps^\star
\quad\text{where~}Y \sim \mathcal{N}(0,1).
\end{equation}
Moreover, if $\eps^\star$ is of the form $1/2^{\text{\tiny{integer}}}$
and $|\ell|,|m| \leq L$, then  the running time is
$\tilde{O}(1/(\eps^\star)^4 + {\log(L)}/\eps^\star)$  and
each element $s_i$ is of the form $a_i/(C' / (\eps^\star)^2)$
where $C'$ is a (positive integer) absolute constant and
$a_i$ is an integer satisfying $|a_i| = \tilde{O}(L/(\eps^\star)^2).$
\end{lemma}

\begin{proof}
We will use the following basic fact which says that it is easy to
construct a high-accuracy $\eps$-cover for $\mathcal{N}(0,1)$ w.r.t.
Kolmogorov distance:

\begin{fact} \label{fact:cover}
There is a deterministic procedure with the following property:  given as
input any value $\eps^\star=1/2^j$ where $j \geq 0$ is an integer,
the procedure runs in time $\tilde{O}(1/(\eps^\star)^4)$
bit operations
and outputs a set $T=\{t_1,\dots,t_{4/\eps^\star}\}$ of $4/\eps^\star$
real values such that
$\dk({\cal D}_T,Z) \leq \eps^\star$ where $Z \sim \mathcal{N}(0,1).$
Moreover each $t_i$ is a rational number of the form
integer$/(C/\eps^\star)$, where $C>0$ is some absolute constant
and the numerator is at most $\tilde{O}(1/\eps^\star).$
\end{fact}
\begin{proof}
(A range of different proofs could be given for this fact; we chose this one
both for its simplicity of exposition and because the computational overhead
of the dynamic programming routine outweighs the runtime of this procedure,
so its exact asymptotic running time is not too important.)

The deterministic procedure is very simple:  it explicitly
computes the values $p_j := {n \choose j}/2^n$ for $j=0,1,\dots,n$ where
 $n = \Theta(1/(\eps^\star)^2)$ is odd (we can and do take it to
additionally be a perfect square).
It is easy to see that all $n$ of these
values can be computed using a total of
${\tilde{O}(n^2) =}  \tilde{O}(1/(\eps^\star)^4)$ bit operations (each of the ${n =}O(1/(\eps^\star)^2)$
binomial coefficients can be
computed from the previous one by performing a constant number of
multiplications and divisions of an ${n =}O(1/(\eps^\star)^2)$-bit number
with an $O(\log(1/\eps^\star))$-bit number).
{For $r \in \mathbb{R}$, let $A_r = \{z \in \mathbb{Z} : r \ge z \ge 0 \}$ and let $P_r = \littlesum_{z \in A_r} p_z$.
The Berry-Ess{\'e}en
theorem implies that for all $r=0,1,\dots,n$ we have
\[
|P_r -
\Pr_{Z \sim \mathcal{N}(0,1)}[Z \leq ((r-(n+1)/2)/ \sqrt{n})]| \leq 1/\sqrt{n}
\leq \eps^\star/10.
\]
Let $\mathcal{D}_{\mathcal{B}}$ denote the binomial distribution $B(n,1/2)$.
Consider the distribution $\widetilde{\mathcal{D}_{\mathcal{B}}} = (\mathcal{D}_{\mathcal{B}} -(n+1)/2)/\sqrt{n}$. Then, note that $d_K(\mathcal{N}(0,1), \widetilde{\mathcal{D}_{\mathcal{B}}} ) \le \eps^{\ast}/10$. Note that
every element in the support of $\widetilde{\mathcal{D}_{\mathcal{B}}}$ is a rational number whose numerator is an integer (bounded by $C/\eps^{\ast 2}$) and the denominator is $C/\eps^{\ast}$. Further, as $1/\sqrt{n} \leq \eps^\star/10$ for any $z \in \mathbb{R}$, $\Pr[\widetilde{\mathcal{D}_{\mathcal{B}}}=z] \le \eps^{\ast}/10$. This gives a straightforward method to obtain a distribution $\mathcal{D}'$ supported on a $2/\eps^{\ast}$ points
such that  $d_K(\mathcal{D}', \widetilde{\mathcal{D}_{\mathcal{B}}}) \le \eps^{\ast}/2$.
Rounding each value to a multiple of $\eps^{\ast}/4$ (and carefully moving the mass around),  it is easy to obtain the set $T$ of size $4/\eps^{\ast}$ points such that
$d_{TV}(\mathcal{D}', \mathcal{D}_T) \le \eps^{\ast}/4$. As a consequence,
$$
d_K(\mathcal{D}_T, \mathcal{N}(0,1)) \le d_{TV}(\mathcal{D}', \mathcal{D}_T) + d_K(\mathcal{D}', \widetilde{\mathcal{D}_{\mathcal{B}}})+ d_K(\mathcal{N}(0,1), \widetilde{\mathcal{D}_{\mathcal{B}}} )  \le \eps.
$$
The claim about the representation of points in $T$ follows from the fact that they are a subset of the points in the support of $\widetilde{\mathcal{D}_{\mathcal{B}}}$ for which we proved this property.
}
This concludes the proof of Fact~\ref{fact:cover}
\end{proof}


We next make the following claim for Kolmogorov distance for
functions of random variables which is an analogue of the
``data processing inequality" for total variation distance.
First, we make the following definition.
\begin{definition}\label{def:unimodal}
Let $f : \mathbb{R} \rightarrow \mathbb{R}$ be a differentiable function. It is said to be $k$-modal if there are at most $k$ points $z_1, \ldots, z_k$
such that $\frac{df}{dx}|_{x=z_i} =0$.
\end{definition}
\begin{claim}\label{clm:data-processing}
Let $f$ be a unimodal function and let $X$ and $Y$ be real valued random variables. Then $d_K(f(X), f(Y)) \le 2 d_K(X,Y)$.
\end{claim}
\begin{proof}
For any real number $t$, there are two possibilities :
\begin{itemize}
\item[(i)] There are real numbers $t_1 \le t_2$ such that $f(x) \ge t $ if and only if $x \in [t_1, t_2]$.
\item[(ii)] There are real numbers $t_1\le t_2$ such that $f(x) \ge t$ if and only if $x \not \in (t_1,t_2)$.
\end{itemize}
In case (i),
$$
\Pr[f(X) \in [t,\infty)] = \Pr[ X \in [t_1, t_2]]  \quad
\text{and}
\quad \Pr[f(Y) \in [t,\infty)] = \Pr[ Y \in [t_1, t_2]].
$$
As a consequence,
\begin{eqnarray*}
d_K(f(X), f(Y)) &=& \sup_{t \in \mathbb{R}}  \left| \Pr[f(X)
\in [t,\infty)]  - \Pr[f(Y) \in [t,\infty)] \right| \\
&\le& \sup_{t_1 , t_2 \in \mathbb{R}} \left| \Pr[ X \in [t_1, t_2]]
-\Pr[ Y \in [t_1, t_2]]\right| \le 2d_K(X,Y).
\end{eqnarray*}
In case (ii),
$$
\Pr[f(X) \in [t,\infty)] = \Pr[ X \not \in [t_1, t_2]]  \quad
\text{and}
\quad \Pr[f(Y) \in [t,\infty)] = \Pr[ Y \not \in [t_1, t_2]]  .
$$
As a consequence,
\begin{eqnarray*}
d_K(f(X), f(Y)) &=& \sup_{t \in \mathbb{R}}  \left| \Pr[f(X) \in [t,\infty)]
- \Pr[f(Y) \in [t,\infty)] \right| \\ &\le& \sup_{t_1 , t_2 \in \mathbb{R}}
\left| \Pr[ X \not \in [t_1, t_2]]-\Pr[ Y \not \in [t_1, t_2]]\right|
\le 2d_K(X,Y).
\end{eqnarray*}
This proves the stated claim.
\end{proof}

We first apply Fact~\ref{fact:cover} to construct a
distribution   ${\cal D}_T$  such that $d_K( {\cal D}_T, \mathcal{N}(0,1))
\le \eps^{\ast}/2$. We then observe that the function $f(t) = \ell t^2
+ m t$ is unimodal and hence ${\cal D}_S = f({\cal D}_T)$,
we have $d_K({\cal D}_S, \ell Y^2 + m Y) \le \eps^{\ast}$
where $Y \sim \mathcal{N}(0,1)$.
This concludes the proof of Lemma~\ref{lem:discretize1}.
\end{proof}


With Lemma~\ref{lem:discretize1} in hand it is simple to obtain the
following:

\begin{lemma} \label{lem:subadditive}
Let $X_1,\dots,X_K$ be independent random variables where
$X_i \sim {\cal D}_{S_i}$
(see Step~2 of {\tt Count-Junta}).  Then
\begin{equation}
\label{eq:sum}
\dk \left( \littlesum_{i=1}^K X_i, \littlesum_{i=1}^K (\lambda'_i y_i^2 + \mu'_i y_i) \right)
\leq K \eps^\star \leq \eps/2
\quad \text{where~}y=(y_1,\dots,y_k) \sim \mathcal{N}(0,1)^K.
\end{equation}
\end{lemma}
\begin{proof}
This follows immediately from (\ref{eq:DS}) and
the sub-additivity property of
Kolmogorov distance:  for $A_1,\dots,A_n$
independent random variables and $B_1,\dots,B_n$ independent
random variables,
\[
\dk( \littlesum_{i=1}^n A_i,  \littlesum_{i=1}^n B_i)
\leq  \littlesum_{i=1}^n \dk(A_i,B_i).
\]
(See e.g. Equation~(4.2.3) of~\cite{BK01} for an explicit
statement; this also follows
easily from the triangle inequality and the basic bound that
$\dk(X_1 + Y, X_2 + Y) \leq \dk(X_1,X_2)$ for
$X_1,X_2,Y$ independent random variables.)
\end{proof}


Finally we turn to Step~3, the dynamic programming.
The algorithm {\tt DP} uses dynamic programming to compute the
exact value of
$\Pr[X_1 + \cdots + X_K + {\tau'} \geq 0]$, where $X_1,\dots,X_K$
are independent random variables with $X_i$ distributed
according to ${\cal D}_{S_i}$.  Observe that by Lemma~\ref{lem:discretize1},
for any $1 \leq i \leq K$
the partial sum $X_1 + \cdots + X_i$ must always be of the form
$c/(C'/(\eps^\star)^2)$ where $c$ is an integer satisfying
$|c| \leq \tilde{O}((i/\eps)/(\eps/K)^2) = \tilde{O}((K/\eps)^3)
= N$, where $N={\poly(1/\eps)}.$
Thus the dynamic program has a variable $v_{i,n}$
for each pair $(i,n)$ where $1 \leq i \leq K$ and $|n| \leq N$;
the value of variable $v_{i,n}$ is $\Pr[X_1 + \cdots + X_i =
n/(C'/(\eps^\star)^2).$
Given the values of variables $v_{i-1,n}$ for all $n$ and the multiset $S_i$
it is straightforward to compute the values of variables $v_{i,n}$ for
all $n$.
Since each nonzero probability under any distribution ${\cal D}_{S_i}$
is a rational number {with both numerator and denominator
$O(\log(1/\eps))$ bits long}, the bit
complexity of every value $v_{i,n}$ is at most $\tilde{O}(K)$ bits,
and since there are $K N$ entries in the table,
the overall running time of {\tt DP} is $\tilde{O}(K^2 N) =
{\poly(1/\eps)}$ bit operations.  The procedure {\tt DP}$(S_1,\dots,
S_K,{\tau'})$
returns the value $v = \sum_{n=0}^N v_{K,n}$, which by the above
discussion is exactly equal to

\begin{equation} \label{eq:sumofXis}
\Pr_{(X_1,\dots,X_K) \sim {\cal D}_{S_1} \times \cdots \times {\cal D}_{S_K}}
[X_1 + \cdots + X_K + {\tau'} \geq 0].
\end{equation}

Now equations~(\ref{eq:round-q}) and~(\ref{eq:sum}) together establish
that the value (\ref{eq:sumofXis})
output by {\tt Count-Junta} satisfies (\ref{eq:close})
as required for correctness.  This concludes the proof of
Theorem~\ref{thm:count-junta}. \qed

\subsubsection{Putting it all together.}

\fi

Combining algorithms {\tt Construct-Junta-PTF} and {\tt Count-Junta},
Theorems~\ref{thm:construct-junta-PTF} and~\ref{thm:count-junta},
give our main result for degree-2 PTFs over Gaussian
variables:

\begin{theorem} \label{thm:count-gaussian-ptf}
[Deterministic approximate counting of degree-$2$ PTFs over Gaussians]
There is an algorithm with the following properties:
It takes as input an explicit description of an $n$-variable degree-$2$
polynomial $p$ with {$b$-bit} integer coefficients and a value $\eps > 0.$
It runs (deterministically) in time $\poly(n,b,1/\eps)$
and outputs a value $v \in [0,1]$ such that
\ifnum\confversion=0
\begin{equation} \label{eq:close2}
\left|\Pr_{x \sim \mathcal{N}(0,1)^n}[p(x) \geq 0] - v \right| \leq \eps.
\end{equation}
\else
$\left|\Pr_{x \sim \mathcal{N}(0,1)^n}[p(x) \geq 0] - v \right| \leq \eps.$
\fi
\end{theorem}

\section{Deterministic approximate counting for degree-2 polynomials
over $\{-1,1\}^n$}
\label{sec:count-bool}

In this section we extend the results of the previous section to give
a deterministic algorithm for approximately counting
satisfying assignments of a degree-2 PTF over the Boolean hypercube.
We prove the following:

\begin{theorem} \label{thm:count-boolean-ptf}
[Deterministic approximate counting of degree-2 PTFs over
$\{-1,1\}^n$]
There is an algorithm with the following properties:
It takes as input an explicit description of an $n$-variable degree-2
multilinear polynomial $p$ with {$b$-bit} integer coefficients
\ignore{and total bit complexity $|p|$,} and a value $\eps > 0.$
It outputs a value $v \in [0,1]$ such that
\ignore{\begin{equation} \label{eq:close3}}$
\left|\Pr_{x \sim \{-1,1\}^n}[p(x) \geq 0] - v \right| \leq \eps
$
and runs (deterministically) in time
$\poly(n,b,2^{\tilde{O}(1/\eps^{{9}})}).$
\end{theorem}

The main ingredient in the proof of Theorem~\ref{thm:count-boolean-ptf}
is the ``regularity lemma for PTFs'' of \cite{DSTW:10}.
As stated in \cite{DSTW:10},
this lemma is an existential
statement which says that every degree-$d$ PTF over $\{-1,1\}^n$ can
be expressed as a shallow decision tree with variables at the internal
nodes and degree-$d$ PTFs at the leaves, such that a random
path in the decision tree is quite likely to reach a leaf that has a
``close-to-regular'' PTF.  \ifnum \confversion=1
Fortunately the \cite{DSTW:10} proof is constructive and in fact yields
the following statement:
\begin{theorem} \label{thm:algorithmic-regularity}
Let $p(x_1,\dots,x_n)$ be an input multilinear degree-$d$ PTF with
{$b$-bit} integer coefficients.
Fix any $\tau>0$.  There is an algorithm $A_{\mathrm{Construct-Tree}}$
which, on input $p$ and a parameter $\tau > 0$, runs in
{$\poly(n,b,2^{\depth(d,\tau)})$}
time and outputs a decision tree $\T$ of depth
$
\depth(d,\tau) := {\frac 1 \tau} \cdot \left(d \log {\frac 1 \tau}
\right)^{O(d)},
$
where each internal node of the tree is labeled with a variable and
each leaf $\rho$ of the tree is labeled with a pair $(p_\rho,\mathrm{label}
(\rho))$
where $\mathrm{label}(\rho) \in \{+1,-1,\text{``fail''},\text{``regular''}\}.$
The tree $\T$ has the following properties:

\begin{enumerate}

\item Every input $x \in \{-1,1\}^n$ to the tree reaches a leaf
$\rho$ such that $p(x)=p_\rho(x)$;

\item If leaf $\rho$ has $\mathrm{label}(\rho) \in \{+1,-1\}$ then
$\Pr_{x \in \{-1,1\}^n}[\sign(p_\rho(x)) \neq \mathrm{label}(\rho)]
\leq \tau$;

\item If leaf $\rho$ has $\mathrm{label}(\rho) = \text{``regular''}$
then $p_\rho$ is $\tau$-regular; and

\item With probability at most $\tau$, a random path from the
root reaches a leaf $\rho$ such that $\mathrm{label}(\rho)=\text{``fail''}.
$

\end{enumerate}

\end{theorem}

\fi

 \ifnum \confversion=0The precise statement is:

\begin{lemma} \label{lem:DSTW10} [Theorem~1 of \cite{DSTW:10}]
Let $f(x) = \sign(p(x))$ be any degree-$d$ PTF.  Fix any $\tau > 0.$
Then $f$ is equivalent to a decision tree $\T$, of depth
\[
\depth(d,\tau) := {\frac 1 \tau} \cdot \big(d \log {\frac 1
\tau}\big)^{O(d)}
\]
\noindent with variables at the internal nodes and
a degree-$d$ PTF $f_\rho = \sign(p_\rho)$ at each leaf $\rho$,  with the
following property: with probability at least $1 - \tau$, a random path
from the root reaches a leaf $\rho$ such that $f_\rho$ is $\tau$-close to
some $\tau$-regular
degree-$d$ PTF.
\end{lemma}

Intuitively, this lemma is helpful for us because for regular polynomials
$g$ we can simply use $\Pr_{x \sim \mathcal{N}(0,1)^n}[g(x)\geq 0]$
(which we can approximate efficiently using
Theorem~\ref{thm:count-gaussian-ptf})
as a proxy for $\Pr_{x \sim \{-1,1\}^n}[g(x) \geq 0]$ and incur only small
error.   However, to use the lemma in our context we need a
deterministic algorithm which efficiently constructs the
decision tree.  While it is not clear from the lemma statement above,
fortunately the \cite{DSTW:10} proof in fact provides such an algorithm,
as we explain below.

\subsection{Proof of Theorem~\ref{thm:count-boolean-ptf}}
\ignore{
Observe that each leaf $\rho$ of a decision tree corresponds
to a restriction fixing some coordinates of an input $x$.
We write ``$x \sim \rho$'' to indicate that input $x \in \{-1,1\}^n$
is compatible with the restriction $\rho$.
}

As we describe below, the argument of
\cite{DSTW:10} actually gives the following lemma, which is
an effective version of Lemma~\ref{lem:DSTW10} above.

\begin{theorem} \label{thm:algorithmic-regularity}
Let $p(x_1,\dots,x_n)$ be an input multilinear degree-$d$ PTF with
{$b$-bit} integer coefficients.
Fix any $\tau>0$.  There is an algorithm $A_{\mathrm{Construct-Tree}}$
which, on input $p$ and a parameter $\tau > 0$, runs in
{$\poly(n,b,2^{\depth(d,\tau)})$}
time and outputs a decision tree $\T$ of depth
\[
\depth(d,\tau) := {\frac 1 \tau} \cdot \left(d \log {\frac 1 \tau}
\right)^{O(d)},
\]
where each internal node of the tree is labeled with a variable and
each leaf $\rho$ of the tree is labeled with a pair $(p_\rho,\mathrm{label}
(\rho))$
where $\mathrm{label}(\rho) \in \{+1,-1,\text{``fail''},\text{``regular''}\}.$
The tree $\T$ has the following properties:

\begin{enumerate}

\item Every input $x \in \{-1,1\}^n$ to the tree reaches a leaf
$\rho$ such that $p(x)=p_\rho(x)$;

\item If leaf $\rho$ has $\mathrm{label}(\rho) \in \{+1,-1\}$ then
$\Pr_{x \in \{-1,1\}^n}[\sign(p_\rho(x)) \neq \mathrm{label}(\rho)]
\leq \tau$;

\item If leaf $\rho$ has $\mathrm{label}(\rho) = \text{``regular''}$
then $p_\rho$ is $\tau$-regular; and

\item With probability at most $\tau$, a random path from the
root reaches a leaf $\rho$ such that $\mathrm{label}(\rho)=\text{``fail''}.
$

\end{enumerate}

\end{theorem}

We prove Theorem~\ref{thm:algorithmic-regularity} in
Section~\ref{sec:alg-reg} below, but first we show how
it gives Theorem~\ref{thm:count-boolean-ptf}.
\fi
\ifnum \confversion=1
The proof of Theorem~\ref{thm:algorithmic-regularity}  is deferred to the full version but we show here how it gives
Theorem~\ref{thm:count-boolean-ptf}.
\fi
\medskip

\noindent {\bf Proof of Theorem~\ref{thm:count-boolean-ptf},
assuming Theorem~\ref{thm:algorithmic-regularity}:}  The
algorithm for approximating $\Pr_{x \in \{-1,1\}^n}[
p(x) \geq 0]$ to $\pm \eps$ works as follows.  It
first runs $A_{\mathrm{Construct-Tree}}$
with its ``$\tau$'' parameter set to $\Theta(\eps^{{9}})$
to construct the decision tree $\T$.  It then iterates over
all leaves $\rho$ of the tree.  For each leaf $\rho$ at depth
$d_\rho$ that has $\mathrm{label}(\rho)=+1$ it adds $2^{-d_\rho}$ to
$v$ (which is initially zero), and for each leaf $\rho$ at depth
$d_\rho$ that has $\mathrm{label}(\rho)=\text{``regular''}$
it runs the algorithm of Theorem~\ref{thm:count-gaussian-ptf} on $p_\rho$
(with its ``$\eps$'' parameter set to $\Theta(\eps^{{9}})$) to
obtain a value $v_\rho \in [0,1]$ and adds $v_\rho \cdot 2^{-d_\rho}$ to $v$.
It outputs the value $v \in [0,1]$ thus obtained.

Theorems~\ref{thm:algorithmic-regularity} and~\ref{thm:count-gaussian-ptf}
imply that the running time is as claimed.
To establish correctness of the algorithm we will use the
``invariance principle'' of \cite{MOO10}: 

\begin{theorem}[\cite{MOO10}]
\label{thm:invariance} Let $p(x) = \littlesum_{S\subseteq[n], |S| \leq d}
p_S x_S $ be a degree-$d$
multilinear polynomial with
$\Var[p]=1$.  Then
$
\sup_{t \in \R}|\Pr_{x \in \bn }[p(x)\leq t] -
\Pr_{\mathcal{G} \sim \mathcal{N}(0,1)^n}[p(\mathcal{G})\leq t]| \leq O(d\tau^{1/(
{4d+1})}),$
where $\tau$ is such that
each coordinate $i\in[n]$ has $\Inf_i(p) \leq \tau$.
 \end{theorem}

By Theorem~\ref{thm:algorithmic-regularity}, the leaves of $\T$ that are
marked $+1$, $-1$ or ``fail''
collectively contribute at most $\Theta(\eps^{{9}}) \leq \eps/2$ to the error
of the output value $v$.  Theorem~\ref{thm:invariance} implies that
each leaf $\rho$ at depth $d_\rho$ that is marked ``regular''
contributes at most
$2^{-d_\rho} \cdot \eps/2$ to the error, so the total contribution
from all such leaves is at most $\eps/2$.
This concludes the proof of Theorem~\ref{thm:count-boolean-ptf}.
\qed
\ifnum \confversion=0
\subsection{Proof of Theorem~\ref{thm:algorithmic-regularity}:
The \cite{DSTW:10} construction.} \label{sec:alg-reg}

Theorem~1 of \cite{DSTW:10} establishes the existence of the claimed
decision tree $\T$ by analyzing an iterative procedure that constructs
$\T$.  Inspection
of this procedure reveals that it can be straightforwardly
implemented by an efficient
deterministic algorithm.  We first provide some
details of the procedure below
and then analyze its running time.

The iterative procedure uses parameters $\beta = \tau$,\ignore{$\tilde{\tau}'
= \tau$,} and $\tilde{\tau}$ chosen such that
${\tau} = \tilde{\tau} \cdot (C' d \ln d \ln(1/\tilde{\tau}))^d$
where $C'$ is a universal constant (see Lemma~12 and the proof of Theorem~1
of \cite{DSTW:10}).
It works to construct $\T$ by processing each
node of the tree that has not yet been declared a leaf of ${\cal T}$
in the manner that we now describe.

\medskip

\noindent {\bf Processing a single node:}
Consider a given node that is currently a leaf of the partially-constructed
decision tree but has not yet been declared a leaf of ${\cal T}.$
Call such a node $\rho$; it corresponds to a restriction
of some of the variables, and
such a node is currently labeled with the restricted
polynomial $p_\rho$.  (At the beginning of the procedure the
node $\rho$ is the root of $\T$, corresponding to the empty restriction
that fixes no variables, and the polynomial $p_\rho$ is simply $p$.)
Let us write $p_\rho(x) = \sum_{|S| \leq d, S \subset [n]} p_{\rho,S} x_S$
where $x_S = \prod_{i \in S} x_i$.

If the depth $d_\rho$ of $\rho$ is greater than $\depth(d,\tau)$ then the
procedure declares $\rho$ to be a leaf of $\T$ and labels
it with the pair $(p_\rho,\text{``fail''})$.  Otherwise,
the procedure first computes $\Inf_i(p_\rho) = \sum_{S \ni i} (p_{\rho,S})^2$
for all $i=1,\dots,n$ and $\Inf(p_\rho) = \sum_{i=1}^n \Inf_i(p_\rho)$.
It sorts the variables in decreasing
order of influence (for notational convenience we shall suppose that
$\Inf_1(p_\rho) \geq \Inf_2(p_\rho) \geq \cdots$), and operates as follows:

\begin{enumerate}
\item
If $\Inf_1(p_\rho) \leq \tau \cdot \Inf(p_\rho)$ then
the node $\rho$ is declared a leaf of $\T$ and is labeled with the pair
$(p_\rho,\text{``regular''}).$

Otherwise,
let $\ci_\tau(p_\rho)$, the $\tau$-critical index of $p_\rho$, be the
least $i$ such that
\[
\Inf_{i+1}(p_\rho) \leq \tau \cdot \sum_{j=i+1}^n \Inf_j(p_\rho).
\]
Let $\alpha = \Theta(d \log \log (1/\tau) + d \log d)$.

\item
If $\ci_\tau(p_\rho) \geq \alpha/\tilde{\tau}$ then the procedure ``expands''
node $\rho$ by replacing it with a complete decision tree of
depth $\alpha/\tilde{\tau}$, where all internal nodes at the $i$-th level
of this tree contain variable $x_i$.  For each new restriction
$\rho'$ (an extension of $\rho$)
resulting from this expansion the procedure computes
$p_{\rho'}$ and labels node $\rho'$ with that polynomial.

Let us write $p_\rho(x) = p_\rho(x_H,x_T) =
p'_\rho(x_H) + q_\rho(x_H,x_T)$ where $p'_\rho(x_H)$ is the truncation
of $p$ containing only the monomials all of whose variables
lie in the set $H=\{1,\dots,\ci_\tau(p_\rho)\}$.

It is easy to see that the constant term of the polynomial $p_{\rho'}$
is precisely $p'_\rho(\rho')$.
The procedure computes
$|p'_\rho(\rho')|$ and
$\|q_\rho(\rho,x_T)\|_2$.  If
$|p'_\rho(\rho')| \geq t^\star \eqdef 1/(2C^d)$
(here $C>0$ is a universal constant; see Definition~2
and the discussion at the end of Section~1.2 of \cite{DSTW:10})
and
$\|q_\rho(\rho,x_T)\|_2 \leq t^\star \cdot
(\Theta(\log (1/\beta)))^{-d/2}$ then the procedure
declares $\rho'$ to be a leaf and labels it with the pair
$(p_{\rho'},\sign(p'_\rho(\rho'))).$

\item If $\ci_\tau(p_\rho) < \alpha/\tilde{\tau}$ then the procedure expands
node $\rho$ by replacing it with a complete decision tree of
depth $\ci_\tau(p_\rho)$, where again all
internal nodes at the $i$-th level of this tree contain
variable $x_i$.
As in the previous case,
for each new restriction
$\rho'$ resulting from this expansion the procedure computes
$p_{\rho'}$ and labels node $\rho'$ with that polynomial.

\end{enumerate}

It is clear that the above procedure constructs a tree $\T$ that satisfies
properties (1), (2) and (3) of Theorem~\ref{thm:algorithmic-regularity}.
The analysis of \cite{DSTW:10} establishes that the tree $\T$
satisfies property (4).

Finally, the running time bound is easily verified from the
description of the algorithm and the fact that the input
is a degree-$d$ PTF with $b$-bit integer coefficients.
\ignore{
{Bound bit complexity
of the polynomials $p_\rho$ that are ever constructed by this procedure --
using depth bound this should be easy.
Bound runtime required to compute influences etc. at each node.
Should be straightforward.}
}

\subsection{Fully polynomial deterministic approximate counting
for regular degree-2 PTFs.}

As a special case of the above analysis we easily obtain
the following result for regular PTFs:

\begin{theorem} \label{thm:count-regular-boolean-ptf}
[Deterministic approximate counting of regular degree-2 PTFs over
$\{-1,1\}^n$]
Let $p$ be an $n$-variable degree-2
multilinear polynomial $p$ with {$b$-bit} integer coefficients
that is $O(\eps^{{9}})$-regular.
Then the above algorithm runs in deterministic time
$\poly(n,b,1/\eps)$
and outputs a value $v \in [0,1]$ such that
\[
\left|\Pr_{x \sim \{-1,1\}^n}[p(x) \geq 0] - v \right| \leq \eps.
\]
\end{theorem}

This is because if $p$ is already $\eps^{{9}}$-regular then the
tree-construction procedure will halt immediately at the root.

\fi

\section{A deterministic FPT
approximation algorithm for absolute moments} \label{sec:moments}

In this section we prove Theorem~\ref{thm:compute-kth-moment}.
{Note that since we are working with polynomials over the
domain $\{-1,1\}^n$, it is sufficient to consider
multilinear polynomials. }

We begin with the following easy observation:

\begin{observation} \label{obs:moments-big}
Let $q(x)$ be a degree-2 multilinear
polynomial over $\{-1,1\}^n$ that has $\E_{x \in \{-1,1\}^n}[q(x)^2] = 1.$
Then for all $k \geq 1$ we have that the $k$-th raw moment
$\E_{x \in \{-1,1\}^n}[|q(x)|^k]$ is at least $c$ where $c>0$ is some universal constant.
\end{observation}
\begin{proof}
For $k \geq 2$ this is an immediate consequence of the monotonicity
of norms, which gives us
\[
1 = \E[|q(x)|^2]^{1/2} \leq \E[|q(x)|^k]^{1/k} \quad \text{for~}k\geq 2.
\]
For $k=1$ the desired statement is an easy consequence of
Theorem~4.1 of \cite{AH11}.
\end{proof}

Given an input degree-2 {multilinear}
polynomial $p(x_1,\dots,x_n)$, we may divide
by $\|p\|_2$ to obtain a scaled version $q=p/\|p\|_2 $ which has
$\|q\|_2=1.$  Observation~\ref{obs:moments-big} implies that
an additive $\pm \eps$-approximation to
$\E[|q(x)|^k]$ is also a multiplicative $(1\pm O(\eps))$-approximation
to $\E[|q(x)|^k]$.
Multiplying the approximation by $\|p\|_2^k$ we obtain a multiplicative
$(1 \pm O(\eps))$-approximation to $\E[|p(x)|^k]$.  Thus to
prove Theorem~\ref{thm:compute-kth-moment} it suffices to give a
deterministic algorithm which finds an additive $\pm \eps$-approximation
to $\E[|q(x)|^k]$ for degree-2 polynomials with $\|q\|_2=1$.
We do this by proving Theorem~\ref{thm:deg2-computemoments} below:

\begin{theorem} \label{thm:deg2-computemoments}
Let $p(x)$ be an input multilinear degree-2 PTF with {$b$-bit}
integer coefficients.
Let $q(x) = p(x)/\|p\|_2$ so $\|q\|=1.$

There is an algorithm $A_{\mathrm{moment}}$ that, on input
$k \in \Z^+$, $p$, and $\eps > 0$, runs in time
{$\poly \left( n,b,2^{\tilde{O}((k \log k \log(1/\eps))^{9k}/\eps^9)}
\right)$}
and outputs a value $\tilde{\mu}_k$ such that
\[
\left|\tilde{\mu}_k - \E_{x \in \{-1,1\}}[|q(x)|^k] \right|
\leq \eps.
\]
\end{theorem}

\subsection{Proof of Theorem~\ref{thm:deg2-computemoments}.}
The idea behind the theorem is very simple.  Since we can estimate
$\Pr_{x \sim \{-1,1\}^n}[q(x) \geq t]$ for any $t$ of our choosing, we can get
a detailed picture of where the probability mass of the random variable
$q(x)$ lies (for $x$ uniform over $\{-1,1\}^n$), and with this detailed
picture it is straightforward to estimate the $k$-th moment.

We now enter into the details.
For $j \in \Z$ let $q_{j,\Delta}$ denote
$\Pr_{x \in \{-1,1\}^n}[q(x) \in [(j-1)\Delta,j\Delta]]$.

We start with the following claim which follows immediately from
Theorem~\ref{thm:count-boolean-ptf}:
\begin{claim} \label{claim:interval}
Fix any $0<\Delta<1$ and any degree-2 multilinear
polynomial ${p}$ {with $b$-bit integer coefficients.
As above let $q(x)=p(x)/\|p\|_2.$}
There is a
{$\poly(n,b,2^{\tilde{O}(1/\eps^{{9}})})$}-time
algorithm which, given as input
${p}$, $0<\eps<1/2$, $\Delta \in \R$ and $j \in \Z$,
outputs a value $\tilde{{q}}_{j,\Delta}$ such that
\[
\tilde{{q}}_{j,\Delta} \in
[{q}_{j,\Delta} - \eps, {q}_{j,\Delta}+\eps].
\]
\end{claim}
\ignore{
(Note that we can assume
the bit complexity of $\eps$ is dominated by the value of $1/\eps$;
if $\eps$ is say in $[0.1,0.2]$ but has some \anote{I removed the word annoying :-)} huge bit complexity we
could just round it down by at most a factor of 2 and get small bit complexity.
The $\log(j)$ runtime dependence is just for computing the value $j \Delta$.)
\rnote{Say the things in this paragraph better.}
}

We recall the following tail bound for polynomials in $\{-1,1\}$
random variables {which follows easily from
Theorem~\ref{thm:dcb}:}

\begin{theorem} \label{thm:gaussian-tail-bound}
Let $q$ be a degree-2 polynomial {with $\|q\|_2=1.$  }
For any $z \geq 0$ we have
\[
\Pr_{x \in \{-1,1\}^n}[|{q}(x)| \geq z ]
\leq O(1) \cdot \exp(-\Omega(z)).
\]
\end{theorem}

Fix $\Delta>0$.  Let $\gamma_q(t)$ denote the probability mass
function of $q(x)$ when $x$ is distributed uniformly over
$\{-1,1\}^n.$  We may write the $k$-th absolute moment
as
\begin{equation}
\E_{x \in \{-1,1\}^n}[|q(x)|^k]
=
\int_{-\infty}^\infty |t|^k \gamma_q(t) dt.
\end{equation}

For $R\geq 1$ we have
\[
\int_{R-1}^{R} |t|^k \gamma_q(t) dt
=
\int_{(R-1}^{R} t^k \gamma_q(t) dt
\leq
R^k \Pr_{x \in \{-1,1\}^n}[q(x)\geq R-1] \leq O(R^k e^{-\Omega(R)}),
\]
so for integer $M \geq 1$ we have
\[
\int_{t=M}^{\infty} |t|^k \gamma_q(t) dt \leq
\sum_{R=M}^{\infty} O(R^k e^{-\Omega(R)})
\]
which is at most $\eps/8$ for $M = O(k \log k \log {\frac 1 \eps})$
(fix $M$ to this value).
Doing similar analysis for $R \leq -1$ gives that

\[
\E_{x \in \{-1,1\}^n}[|q(x)|^k]
= \int_{-M}^{M}
|t|^k \gamma_q(t) dt \pm \eps' \quad
\text{where~}\eps'<\eps/4.
\]
So to approximate $\E_{x \sim \{-1,1\}^n}[|q(x)|^k]$ to an additive $\pm \eps,$
it suffices to approximate
$
\int_{-M}^M
|t|^k \gamma_q(t) dt$ to an additive $\pm 3\eps/4.$

Fix $\Delta = (\eps/4)^{1/k} \tau / k,$ and consider the
interval $[(j-1) \Delta, j \Delta]$
where $j\geq k/\tau$ for some $0<\tau<1.$
Recalling that
$q_{j,\Delta} =
\Pr_{x \in \{-1,1\}^n}[q(x) \in [(j-1)\Delta,j\Delta]]$, we have
\[
\int_{(j-1)\Delta}^{j \Delta} |t|^k \gamma_q(t) dt \in
[((j-1)\Delta)^k q_{j,\Delta},(j\Delta)^k q_{j,\Delta}] =
q_{j,\Delta} \cdot \Delta^k \cdot [(j-1)^k,j^k].
\]
Since
\[
j^k - (j-1)^k = j^k \left(1 - \left(1 - {\frac 1 j} \right)^k \right)
\leq j^k \left(1 - \left(1 - {\frac \tau k}\right)^k \right) \leq \tau  j^k,
\]
we have
\begin{equation} \label{eq:piece}
\int_{(j-1)\Delta}^{j \Delta} |t|^k \gamma_q(t) dt \in
[1-\tau,1]\cdot |j \Delta|^k q_{j,\Delta}.
\end{equation}
A similar analysis gives that we likewise have
\begin{equation} \label{eq:piece2}
\int_{-j\Delta}^{(-j+1) \Delta} |t|^k \gamma_q(t) dt \in
[1-\tau,1]\cdot |j \Delta|^kq_{j,\Delta}.
\end{equation}
Finally, we observe that
\begin{equation} \label{eq:nearzero}
\int_{-(k/\tau - 1)\Delta}^{(k/\tau - 1) \Delta} |t|^k \gamma_q(t) dt
< ((k/\tau)\Delta)^k = \eps/4,
\end{equation}
where we used $\Delta = (\eps/4)^{1/k} \tau / k$ for the final step.

With the above ingredients in hand it is clear how we shall deterministically
estimate the $k$-th moment $\E_{x \in \{-1,1\}^n}[|q(x)|^k]$.
Given as input an integer
$k \geq 1$, a real value $0 < \eps < 1$, and a degree-2 multilinear
polynomial $q$ with $\|q\|=1$, the algorithm
for estimating this moment works as follows:

\begin{enumerate}

\item Set $M=O(k \log k \log {\frac 1 \eps})$, set
$\tau = \eps/(4M^k)$,
and set
$\Delta = 1/2^r$ where $r$ is the largest value such that
$1/2^r \leq (\eps/4)^{1/k} \tau / k.$
\ignore{\rnote{this is just so
$\Delta$ has small bit complexity; make this clear in some not too
obtrusive way}
}

\item  For $j=(k/\tau - 1)$ to $M/\Delta$:  compute
a {$\pm \tau/4$}-accurate
additive estimate $\tilde{q}_{j,\Delta}$
of $q_{j,\Delta}$ (using Claim~\ref{claim:interval})
and sum the values $|j \Delta|^k \tilde{q}_{j,\Delta}$ to obtain $E_+$.

Similarly, for $j=-(k/\tau - 1)$ to {$-M/\Delta$}:  compute a
{$\pm \tau/4$}-accurate
additive estimate $\tilde{q}_{j,\Delta}$
of $q_{j,\Delta}$ (using Claim~\ref{claim:interval})
and sum the values $|j \Delta|^k \tilde{q}_{j,\Delta}$ to obtain $E_-$.

\item Output $E_+ + E_-.$

\end{enumerate}

It is easy to see that the above algorithm runs in time
{
\[
(M/\Delta) \cdot \poly(n,b, 2^{\tilde{O}(1/\tau^{9})})
=
\poly \left( n,b,2^{\tilde{O}((k \log k \log(1/\eps))^{9k}/\eps^9)}\right)
.
\]
}
To prove correctness recall that we need to show that $E_+ + E_-$
is within $\pm 3\eps/4$ of
$
\int_{-M}^M
|t|^k \gamma_q(t) dt.$
Recalling (\ref{eq:nearzero}), it suffices to show that
$E_+$ and $E_-$ are each within $\pm \eps/4$ of
$
\int_{-M}^{-(k/\tau - 1)\Delta}
|t|^k \gamma_q(t) dt$ and $\int_{(k/\tau - 1)\Delta}^M
|t|^k \gamma_q(t) dt$ respectively.
Recalling our choice of $\tau$, it
follows easily from (\ref{eq:piece}) that \[
\left| E_+ - \int_{-M}^{-(k/\tau - 1)\Delta} |t|^k \gamma_q(t) dt
\right| \leq \tau M^k= \eps/4.
\]
An identical argument works for $E_-$ and the other integral, and we are done
with the proof. \qed

\bibliography{allrefs}
\bibliographystyle{alpha}

\appendix

\section{Definitions and Background} \label{sec:background}
In this section we record the preliminaries we will need.

\subsection{Basic Linear Algebra Facts.}

In this section we record some basic facts from linear algebra
that will be crucial for our proofs.

\begin{definition}(orthogonal matrix)
A matrix $Q \in \R^{n \times n}$ is said to be {\em orthogonal}
if both its columns and its rows comprise a set of
$n$ orthonormal unit vectors.  Equivalently, a matrix $Q \in \R^{n \times n}$
is orthogonal if its transpose is equal to its inverse, i.e., $Q^{T} = Q^{-1}$.
\end{definition}

\begin{theorem}(Spectral Theorem)
If $A \in \R^{n \times n}$ is symmetric, there exists an orthogonal $Q \in \R^{n \times n}$ and a diagonal matrix $\Lambda \in \R^{n \times n}$
such that $A = Q\Lambda Q^T$. The diagonal entries of $\Lambda$ are the eigenvalues of $A$ and the columns of $Q$
are the corresponding eigenvectors. That is, we can write $\Lambda = \mathrm{diag}(\lambda_1, \ldots, \lambda_n)$, $Q = [u^{(1)} \mid \ldots \mid u^{(n)}]$, with
$u^{(i)} \cdot u^{(j)} = \delta_{ij}$,
and $Au^{(i)} = \lambda_i u^{(i)}$ for all $i \in [n]$. The expression $A = Q\Lambda Q^T$ of a symmetric matrix in terms
of its eigenvalues and eigenvectors is referred to as {\em the spectral decomposition of $A$}.
\end{theorem}

\begin{definition}\label{def:quadratic-form}
Given a degree-$2$ polynomial $p : \mathbb{R}^n \rightarrow \mathbb{R}$
defined as $p(x) = \littlesum_{1 \le i \le j \le n} a_{ij} x_i x_j
+ \littlesum_{1 \le i \le n} b_i x_i + C$, we define the (symmetric)
matrix $A$ corresponding to its quadratic part as : $A_{ij} = a_{ij}
(1/2 + \delta_{ij}/2)$.
With this definition, it is easy to see that
$x^T \cdot A \cdot x = \littlesum_{1 \le i \le j \le n} a_{ij} x_i x_j$
for the vector $x = (x_1, \ldots, x_n)$.
\end{definition}

{Throughout the paper we adopt the convention that the eigenvalues
$\lambda_1,\dots,\lambda_n$ of a real symmetric matrix $A$
satisfy $|\lambda_1| \geq \cdots \geq |\lambda_n|$.
We sometimes write $\lambda_{\max}(A)$ to denote $\lambda_1$,
and we sometimes write $\lambda_i(p)$ to refer to the $i$-th eigenvalue
of the matrix $A$ defined based on $p$ as described above.}

\ignore{
\rnote{I ignore'd out the following stuff which seemed redundant
(Fact~\ref{fact:normal-lt} is subsumed by Claim~\ref{clm:poly-equivalence}
and the other things are also covered elsewhere now.  If you agree
please delete this whole comment):

{
\noindent We often represent a homogenous quadratic form $p(x) = \littlesum_{i \le j} a_{i,j}x_ix_j$ by its associated symmetric matrix $A$, where
\begin{displaymath}
A_{i,j} = \left\{ \begin{array}{ll}
a_{i,j}/2, & \textrm{if $i<j$}\\
a_{j,i}/2, & \textrm{if $i>j$}\\
a_{i.j}, & \textrm{if $i=j$}
\end{array} \right.
\end{displaymath}
so that $$p(x) = x^T A x.$$ Let $A = Q \Lambda Q^T$ be the spectral decomposition of $A$.
Then, with the change of variables $y = Q^T x$, we can equivalently write $p$ as
\begin{equation} \label{eqn:cov}
p(x) = y^T \Lambda y = \littlesum_{i=1}^n \lambda_i y_i^2 =  \littlesum_{i=1}^n \lambda_i (u^{(i)} \cdot x)^2.
\end{equation}
We recall the following basic probabilistic fact about Gaussians:
\begin{fact} \label{fact:normal-lt}
Let $x \sim \mathcal{N}(0,1)^n$ and $B \in \R^{n \times n}$ be an
orthogonal matrix. Then $y = Bx$ satisfies $y \sim \mathcal{N}(0,1)^n$.
\end{fact}

}
}
}

\begin{definition}
For a real symmetric matrix $A$, with (real) eigenvalues $\lambda_1, \ldots, \lambda_n$ such that
$|\lambda_1| \ge |\lambda_2| \ge \ldots \ge |\lambda_n|$ we define:
\begin{itemize}
\item The Frobenius norm of $A$ is $\|A\|_F \eqdef \sqrt{\littlesum_{i,j} A^2_{i,j}}.$
\ignore{
\item The sup-norm\rnote{Is this standard terminology?} of $A$ is $\|A\|_{\infty} \eqdef  |\lambda_1|$, i.e., the largest magnitude of an eigenvalue of $A$.
}
\item The trace of $A$ is $\tr(A) \eqdef \littlesum_{i=1}^n A_{ii}$. We have that $\tr(A) = \littlesum_{i=1}^n \lambda_i$.
\end{itemize}
\end{definition}


\medskip

\noindent We recall the following fact:

\begin{fact} \label{fact:lin-alg}
Let $A \in \R^{n \times n}$ be symmetric with eigenvalues $\lambda_1, \ldots, \lambda_n$.
The eigenvalues of the matrix $A^k$ are $\lambda_1^k, \ldots, \lambda_n^k$. Since
$\|A\|_F  = \sqrt{\tr(A^2)}$,  $\| A\|_F = \sqrt{\littlesum_{i=1}^n{\lambda^2_i}}.$
\end{fact}

\subsection{Basic Probabilistic Facts.}

{
Given an $r$-element multiset $S=\{s_1,\dots,s_r\}$ we write ${\cal D}_S$
to denote the distribution which is uniform over the elements of $S$
(so if an element $v$ occurs $j$ times in $S$ we have $\Pr_{x \sim {\cal D}_S}
[x = v] = j/r.$).
}

\begin{fact}\label{fac:variance-difference}
Let $P$ and $Q$ be real valued random variables such that $\Var(P)  = \alpha$ and $\Var(Q) = \eta^2 \alpha$.
Then {$(1+2\eta + \eta^2 ) \alpha \ge \Var(P-Q) \ge (1-2\eta + \eta^2)
\alpha$}.
\end{fact}
\begin{proof}
\begin{eqnarray*}
\Var(P-Q) &=& \mathbf{E} [(P-Q)^2] - (\mathbf{E}[P-Q])^2
= \Var(P) + \Var(Q) - 2 \mathbf{E}[PQ] + 2\mathbf{E}[P] \mathbf{E}[Q]  \\
&=& \Var(P) {+ \Var(Q)} - 2 \Cov(P,Q) \\
&=&  (1 + \eta^2)\alpha - 2 \Cov(P,Q).
\end{eqnarray*}
Now the desired inequalities follow using the
simple inequality $|\Cov(P,Q)| \leq \sqrt{\Var(P)}
\sqrt{\Var(Q)}$ which is a consequence of Cauchy-Schwarz.
\end{proof}

We recall the Berry-Esseen theorem, which states that under suitable
conditions a sum of independent random variables converges (in
Kolmogorov distance) to a normal distribution:

\begin{theorem} \label{thm:be} (Berry-Ess{\'e}en)
Let $\{X_i\}_{i=1}^n$ be a set of independent random variables satisfying
$\E[X_i] = 0$ for all $i \in [n]$, $\sqrt{\littlesum_i \E[X_i^2]} =
\sigma$, and $\littlesum_i \E[|X_i|^3] = \rho_3$.  Let $S = \littlesum_i X_i/\sigma$ and let $F$ denote the cumulative distribution
function (cdf) of $S$. Then
$\sup_x |F(x) - \Phi(x)| \leq \rho_3/\sigma^3$
where $\Phi$ denotes the cdf of the standard gaussian random variable.
\end{theorem}

\begin{fact} \label{fact:gr}
Let $\mu_1, \mu_2 \in \R$ and $0< \sigma_1^2 \le \sigma_2^2$. Then,
$$\dtv\left( \mathcal{N}(\mu_1,  \sigma_1^2), \mathcal{N}(\mu_2,  \sigma_2^2 ) \right) \le \frac{1}{2} \left( \frac{|\mu_1 - \mu_2|}{ \sigma_1} +\frac{ \sigma_2^2  -  \sigma_1^2 }{ \sigma^2_1}\right).$$
\end{fact}

\subsection{Useful Facts about Polynomials.}
We view $\R^n$ as a probability space endowed with the standard $n$-dimensional Gaussian measure.
For a square-integrable function $f:\R^n \to \R$ and $r \ge 1$, we let $\|f\|_r$ denote $(\E_{x \sim \mathcal{N}^n}[ |f(x)^r|])^{1/r}.$
\noindent We will need a concentration bound for low-degree polynomials over independent  Gaussians.

\begin{theorem}[``degree-$d$ Chernoff bound'',  \cite{Janson:97}] \label{thm:dcb}
Let $p: \R^n \to \R$ be a degree-$d$ polynomial. For any
$t > e^d$, we have
\[
\Pr_{x \sim \mathcal{N}(0,1)^n}[
|p(x) - \E[p(x)]| > t  \cdot \sqrt{\Var(p(x))} ] \leq
{d e^{-\Omega(t^{2/d})}}.
\]
{The same bound holds for $x$ drawn uniformly from $\{-1,1\}^n.$}
\end{theorem}

\noindent We will also use the following anti-concentration bound for
degree-$d$ polynomials over Gaussians:

\begin{theorem}[\cite{CW:01}]
\label{thm:cw} Let $p: \R^n \to \R$ be a degree-$d$ polynomial
that is not identically 0.  Then for all $\eps>0$ and all
$\theta \in \R$, we have
\[
\Pr_{x \sim \mathcal{N}(0,1)^n}\left[|p(x) - \theta| < \eps \sqrt{\Var(p)}
\right] \le O(d\eps^{1/d})¬¨¬®‚Äö√Ñ‚Ä†.
\]
\end{theorem}


\begin{definition}
Let $f : \{-1,1\}^n \rightarrow \mathbb{R}$. The influence of the $i^{th}$
coordinate on $f$ under the uniform measure (denoted by $\Inf_i(f)$) is
defined as
$$
\Inf_i(f) = \mathop{\mathbf{E}}_{x_i \in \{-1,1\}} [\Var_{x_1,\ldots,x_{i-1},x_{i+1}, \ldots,x_n \in \{-1,1\}} f(x_1,\ldots,x_n)].
$$
The total influence of a function $f$ (denoted by $\Inf(f)$) is defined as $\sum_{i=1}^n \Inf_i(f)$.
\end{definition}

We now define an extension of the notion of ``critical index" previously used in several works on linear and polynomial threshold functions~\cite{Servedio:07cc, OS11:chow, DRST09}.
\begin{definition} \label{def:ci}
Given a pair of sequences of non-negative numbers $\{c_i\}_{i=1}^n$
and $\{d_i\}_{i=1}^n$ where additionally the sequence $\{c_i\}_{i=1}^n$
is non-increasing, the \emph{$\tau$-critical index} of the pair
is defined to  be the smallest $0 \le i \le n-1$ such that
$$
\frac{c_{i+1}}{\sum_{j>i}( c_j + d_j)} \le \tau.
$$
In case there is no such number, we define the critical index to be $\infty$.
The sequence $\{c_i\}_{i=1}^n$  is called the ``main sequence" and the
sequence $\{d_i\}_{i=1}^n$ is called the ``auxiliary sequence".

\end{definition}
The following is a simple consequence of the definition of critical index.
\begin{fact}\label{fact:critical-index}
Given a pair of sequences of non-negative numbers, $\{c_i\}_{i=1}^n$
and $\{d_i\}_{i=1}^n$, if the $\tau$-critical index of a sequence is $j$,
then $\sum_{i=j+1}^n (c_i+d_i) < (1-\tau)^{j}
\cdot (\sum_{\ell=1}^n c_i +d_i)$.
\end{fact}

As noted earlier, special cases of this definition have appeared in previous
work on polynomial threshold functions.
Below, we recall the notion of the critical index of a polynomial that appeared
previously in the work of Diakonikolas \etal~\cite{DRST09}:
\begin{definition}(\cite{DRST09})
Let $p: \R^n \to \R$ and $\tau>0$. Assume the variables are ordered such that $\Inf_i(p)\geq \Inf_{i+1}(p)$ for all $i \in
[n-1]$.  The {\em $\tau$-critical index} of $f$ is defined to be
the $\tau$-critical index of  the pair of sequences $\{\Inf_i(p)\}_{i=1}^n$
and $\{0 \}_{i=1}^n$ where $\{0 \}_{i=1}^n$ is the auxiliary sequence.
\end{definition}

\ignore{
As a consequence of Fact~\ref{fact:critical-index}, we have the following lemma.
\begin{lemma} \label{lem:exp} (\cite{DRST09})
Let $p: \R^n \to \R$ and $\tau>0$. Let $k$ be the $\tau$-critical index of $p$. For $0\leq j \leq k$ we have
\[
\littlesum_{i=j+1}^n \Inf_i (p) \leq (1 - \tau)^j \cdot \Inf(p).
\]
\end{lemma}

\inote{I guess we do not need this explicitly now, we reprove it in a convolved way in the iterative algorithm.}
\rnote{Should we remove Lemma~\ref{lem:exp}?}

}

\section{Hardness of computing absolute moments}
\label{ap:moment-hard}

In this section, we show that for any fixed odd $k$, it is $\#P$-hard  to 
exactly compute the $k^{th}$ absolute moment of a degree two multilinear 
polynomial with $\{0,1\}$ coefficients over the uniform distribution on 
the hypercube. 

\ignore{In fact, the hardness result also holds if we restrict 
ourselves to polynomials where every variable has influence $o(1)$. 
}

\begin{theorem}\label{thm:hardness-moment}
Given a degree two multilinear polynomial $p: \mathbb{R}^n \rightarrow 
\mathbb{R}$ with $\{0,1\}$ coefficients, it is $\#P$-hard (under Turing 
reductions) to compute $\E_{x \in \{-1,1\}^n}[|p(x)|^k]$,
 the $k^{th}$ absolute moment of $p$ over the 
uniform distribution on $\{-1,1\}^n$, for $k=O(1)$. 
\ignore{Further, the hardness 
result holds for polynomials  where every variable has influence $o(1)$. 
}
\end{theorem}
\begin{proof}
We begin by recalling that given an undirected graph $G=(V,E)$, 
it is $\mathop{NP}$-hard to find the size of the \textsf{MAX-CUT} in $G$. 
In fact, if the size of the \textsf{MAX-CUT} in $G$ is $\nu$, then it 
is $\#P$-hard to find the number of cuts in $G$ whose size is $\nu$ 
(see Papadimitriou~\cite{Papadimitriou}). 

Let $|V|=n$ and $|E|=m$.  We consider the polynomial 
$q_{G,\mathrm{CUT}} : \mathbb{R}^n \rightarrow \mathbb{R}$ defined as 
$q_{G,\mathrm{CUT}}(x) = (|E| - \sum_{\{i,j\} \in E} x_i x_j)/2$;
recall from the introduction that on input $x \in \{-1,1\}^n$ 
the value $q_{G,\mathrm{CUT}}(x)$ equals
the number of edges in the cut corresponding to $x$.
Consequently $q_{G,\mathrm{CUT}}(x))  \in [0,\dots,m]$
for every $x \in \{-1,1\}^n.$
Let $\nu^{\ast}_G = \max_{x \in \{-1,1\}^n} [q_{G,\mathrm{CUT}}(x)]$
denote the size of the \textsf{MAX-CUT} in $G$.
Thus, it is $\#P$-hard to compute the size of the following set
$$
\{x  \in \{-1,1\}^n : q_{g,\mathrm{CUT}}(x) = \nu^{\ast}_G \}.
$$

We next observe that for any fixed $k = O(1)$, there is a poly$(n)$-time
algorithm to compute the $k$-th raw moment $\E_{x \in \{-1,1\}^n}[p(x)^k]$
of a given  degree-2 input polynomial $p(x)$.  The algorithm
works simply by expanding out $p(x)^k$ (in time $n^{O(k)}$),
performing multilinear reduction, and outputting
the constant term; its correctness follows from the fact that $\E_{x \in
\{-1,1\}^n}[x_S]=0$ for every $S \neq \emptyset.$

Suppose $A$ is a $\poly(n)$-time algorithm that, on input
a degree-2 polynomial $p(x)$, outputs the value
$\E_{x \in \{-1,1\}^n}[|p(x)|^k]$.  Given such an algorithm
we can efficiently compute 
$|\{x  \in \{-1,1\}^n : q_{g,\mathrm{CUT}}(x) = \nu^{\ast}_G \}|$
as follows:  For $\ell = m, m-1, \dots$ successively
compute 
$\E[|\ell - q_{G,\mathrm{CUT}}(x)|^k]$ 
(using algorithm $A$) and 
$\E[(\ell - q_{G,\mathrm{CUT}}(x))^k]$ 
(as described above).  Let $\ell^\ast$ be the largest value in $\{m,m-1,
\dots,0\}$ such that 
$
\E[|\ell^\ast - q_{G,\mathrm{CUT}}(x)|^k] \neq 
\E[(\ell^\ast - q_{G,\mathrm{CUT}}(x))^k]
$.
Output the value 
\[
2^{n-1}\left(\E[|\ell^\ast - q_{G,\mathrm{CUT}}(x)|^k]  - 
\E[(\ell^\ast - q_{G,\mathrm{CUT}}(x))^k]\right).
\]

It is clear that the above-described algorithm runs
in poly$(n)$ time.  To verify correctness, first consider
a value of $\ell$ such that $\ell \geq \nu^\ast_G.$
For such an $\ell$ we have that
$
|\ell - q_{G,\mathrm{CUT}}(x)|^k = 
(\ell - q_{G,\mathrm{CUT}}(x))^k
$
for all $x \in \{-1,1\}^n$, and hence the 
raw moment 
$\E[(\ell - q_{G,\mathrm{CUT}}(x))^k]$ will equal the
absolute moment 
$\E[(\ell - q_{G,\mathrm{CUT}}(x))^k]$.
On the other hand, for $\ell = 
\nu^\ast_G - 1$, we have that all cuts of size 
$0,\dots,\nu^\ast_G-1$ contribute the same
amount to $\E[|\ell - q_{G,\mathrm{CUT}}(x)|^k]$
and to $\E[(\ell - q_{G,\mathrm{CUT}}(x))^k]$, but
each cut of size precisely $\nu^\ast_G$ contributes
$1/2^n$ to the absolute moment and $-1/2^n$ to the raw moment.
As a result, we get that 
$2^{n-1}\left(\E[|\ell - q_{G,\mathrm{CUT}}(x)|^k]  - 
\E[(\ell - q_{G,\mathrm{CUT}}(x))^k]\right)$
is  precisely the number of cuts of size
$\nu^\ast_G$, and the theorem is proved.
\end{proof}

\end{document}